\documentclass[11pt]{article}
\usepackage[margin=3cm]{geometry}
\usepackage[utf8]{inputenc}
\usepackage{amsmath}
\usepackage{amsfonts}
\usepackage{amssymb}
\usepackage{amsthm}
\usepackage{mathrsfs}
\usepackage{commath}
\usepackage{latexsym}
\usepackage{accents}
\usepackage{array}
\usepackage{bm}
\usepackage{enumerate}
\usepackage[mathscr]{eucal}
\usepackage{color}
\usepackage{verbatim}
\usepackage[polutonikogreek,english]{babel}
\usepackage{graphicx}
\usepackage{booktabs,caption}
\usepackage{multicol,multirow,tabularx}
\usepackage{siunitx}
\usepackage{relsize}
\usepackage{multirow}
\usepackage{hyperref}
\usepackage{float}
\usepackage{lineno}
\usepackage{cite}
\usepackage{todonotes}
\usepackage{tikz}
\usepackage{psfrag}

\usepackage[square,sort,comma,numbers,compress]{natbib}
\usepackage{overpic}
\usepackage{bbding}
\usepackage{lineno}
\usepackage[dvipsnames]{xcolor}

\newtheorem{theorem}{Theorem}
\newtheorem{lemma}[theorem]{Lemma}
\newtheorem{proposition}[theorem]{Proposition}
\newtheorem{remark}{Remark}

\usetikzlibrary{arrows.meta}
\usetikzlibrary{positioning}
\usetikzlibrary{decorations.pathreplacing,decorations.markings}
\tikzset{->-/.style={decoration={
markings,
mark=at position #1 with {\arrow{>}}},postaction={decorate}}}
\usepackage{pgfplots}
\pgfplotsset{width=7cm,compat=newest,ticks=none}
\usepgfplotslibrary{fillbetween}

\definecolor{royalpurple}{rgb}{0.47, 0.32, 0.66}
\definecolor{pastelgreen}{rgb}{0.47, 0.87, 0.47}
\definecolor{cornellred}{rgb}{0.7, 0.11, 0.11}
\definecolor{pastelorange}{rgb}{1.0, 0.7, 0.28}
\definecolor{darkred}{rgb}{0.55, 0.0, 0.0}
\definecolor{darkpastelgreen}{rgb}{0.01, 0.75, 0.24}
\definecolor{myviolet}{RGB}{169,110,199}
\definecolor{myred}{RGB}{200,0,50}
\definecolor{myblue}{RGB}{0,180,200}
\definecolor{mygreen}{RGB}{0,150,0}

\title{How inertia affects autotoxicity-mediated vegetation dynamics: from close-to to far-from-equilibrium patterns}
\author{Giancarlo Consolo$^{1}$, Carmela Curr\`o$^{1}$, Gabriele Grif\`o$^{2,3,(\ast)}$, Annalisa Iuorio$^{4}$, \\ Giovanna Valenti$^{5}$,
Frits Veerman$^{6}$\\
{$^{1}$\small Department of Mathematical, Computer, Physical and Earth Sciences, University of Messina,} \\
{\small V.le F. Stagno D'Alcontres 31, I-98166 Messina, Italy.}\\
{$^{2}$\small Istituto Nazionale di Alta Matematica ``F. Severi", Piazzale Aldo Moro 5, I-00185 Rome, Italy} \\
{$^{3}$\small Department of Mathematics and Computer Science, University of Palermo,} \\
{\small Via Archirafi 34, I-90123 Palermo, Italy} \\
{$^{4}$\small Department of Engineering, Parthenope University of Naples,}\\ {\small Centro Direzionale - Isola C4, 80143 Naples, Italy}\\
{$^{5}$\small Department of Engineering, University of Messina, C.da di Dio, I-98166 Messina, Italy}\\
{$^{6}$\small Mathematical Institute, University of Leiden, Einsteinweg 55, 2333 CC Leiden, The Netherlands}\\
{\small $^{(\ast)}$Corresponding author: grifo@altamatematica.it}}
\date{}

\begin{document}

\maketitle

\begin{abstract}
In this work, the influence of inertial effects on the formation and evolution of vegetation patterns on sloped arid terrains is investigated in different dynamical regimes ranging from the onset of instability to far-from-equilibrium. Analyses are carried out in a hyperbolic extension of the one-dimensional Klausmeier model, where autotoxicity effects are also taken into account. As the system moves away from the wave bifurcation threshold, two classes of solutions, corresponding to the emergence of qualitatively different coherent structures, arise: small-amplitude periodic migrating bands near onset and large-amplitude travelling pulses in far-from-equilibrium conditions. 
For the first class of solutions, results of linear stability analysis reveal that inertia has a twofold role at onset: it acts as a destabilising mechanism, thereby enlarging the parameter region in which uphill migrating vegetation bands can emerge, and it reduces the pattern migration speed. Its role also manifests itself close to onset, as proved by the Stuart–Landau equation for the pattern amplitude deduced via multiple-scale weakly–nonlinear analysis. Indeed, it is shown that inertial effects may reverse the dynamical regime, from supercritical to subcritical, thus leading to  hysteresis. 
For the second class of solutions, the key qualitative features of travelling vegetation pulses are first captured via numerical simulations and then investigated in detail analytically via Geometric Singular Perturbation Theory (GSPT). The existence of such solutions is proved via construction of the corresponding homoclinic orbits in the associated four-dimensional system. In far-from-equilibrium conditions, inertia is shown to increase pulse speed while preserving the intrinsic multiscale structure of the solution, in full agreement with the numerical findings. Overall, the proposed combined analytical-numerical investigations have depicted several ecological scenarios as a function of the distance from the instability threshold, elucidating that inertia does not exclusively act as a time lag. Indeed, it constitutes an additional degree of freedom which should be considered in shaping vegetation dynamics under environmental stress. \\

\end{abstract}

\section{Introduction}

The ongoing intensification of the global hydrological cycle, driven by anthropogenic climate change, is significantly altering the stability of ecosystems worldwide \cite{UN22}. Among these, arid and semi-arid drylands are particularly vulnerable, as they exist in a precarious balance where water availability is the primary limiting factor for biological growth. In such water-stressed environments, vegetation does not always maintain a uniform cover; instead, it frequently self-organises into striking spatial patterns, ranging from regular stripes and labyrinths to isolated spots and pulses. From an ecological perspective, these patterns are of paramount importance: they represent a resilient response to environmental scarcity, allowing the ecosystem to concentrate limited resources in high-density biomass patches. Consequently, understanding the mathematical principles underlying pattern formation is essential for predicting so-called ``tipping points" and preventing irreversible transitions toward desertification \cite{Bastiaansen2020, Rietkerk2021}.

Traditional modelling efforts, pioneered by the seminal work of Klausmeier \cite{Klausmeier1999}, have successfully captured the essence of these dynamics using reaction-diffusion-advection systems \cite{Hillerislambers2001, Zelnik2013, VanDerStelt2013, Siteur2014, Bastiaansen2019, Eigentler2020, Iuorio2021, Iuorio2023pre, Byrnes2023, Consolo2024}. However, these classical models are typically based on parabolic partial differential equations (PDEs), which implicitly assume an infinite speed of information propagation and an instantaneous response of the vegetation flux to density gradients. In reality, it is well known that any biological system exhibits a certain degree of \textit{inertia}. In the context of ecosystems, the movement and expansion of vegetation are governed by physical constraints and biological time lags, meaning that the biomass flux cannot adapt immediately to environmental shifts. Empirical observations from long-term monitoring of patterned landscapes, particularly those utilising Fourier cross-spectral analysis across diverse dryland regions, have highlighted a significant time lag between environmental forcing and the structural response of the biomass, suggesting that vegetation dynamics are intrinsically governed by inertial effects rather than instantaneous adaptation \cite{Deblauwe2011, Deblauwe2012}. These evidences provide an ecological basis for incorporating inertial terms that account for the ``memory'' and slow successional processes of perennial species. 

This paper introduces a hyperbolic extension of the Klausmeier model for semiarid sloped terrains that explicitly incorporates these inertial effects and explores the structure of the coherent structures observed at, close to and beyond local equilibrium. The aim of this work is to investigate the role played by inertia in shaping the future of dryland landscapes in a changing climate, when the distance from the instability threshold is progressively increased. As it will be here proved, inertial effects act as significant ``transport enhancers", as they do not simply introduce a time lag in the vegetation response but play a multifarious role. They indeed enlarge the width of the parameter region in which migrating patterns can be observed, alter the migration velocity of vegetation bands and pulses, change the dynamical regime of pattern formation. A key finding of this research is that vegetation inertia can promote subcriticality, creating regions of bistability and hysteresis where the ecosystem's state depends heavily on its historical path.

To provide a rigorous and comprehensive description of these complex vegetation dynamics, this work employs a multi-faceted mathematical approach that covers the different classes of solutions obtained within the pattern-forming region: from small-amplitude periodic migrating bands near onset to travelling pulses in far-from-equilibrium conditions. Three tools will be here employed to characterise these solutions: linear stability analysis (LSA), multiple-scale weakly-nonlinear analysis (WNA) and Geometrical Singular Perturbation Theory (GSPT). 

The study begins by investigating the local stability of the spatially homogeneous vegetated state. LSA is crucial for identifying the critical precipitation thresholds and determining how inertia modifies the onset of wave instability, so highlighting the portion of the (rainfall-plant loss) plane where patterned solution may emerge. In the vicinity of the bifurcation threshold, LSA fails to describe the long-term behaviour of the system, so that WNA is applied to derive a Stuart-Landau-type equation for the pattern amplitude. This allows us to investigate the nature of the bifurcation and to distinguish the occurrence of  gradual (supercritical) or abrupt (subcritical) shifts. When the system is driven far-from-equilibrium, it exhibits large-amplitude, localised structures whose core element is represented by travelling pulses. The analysis of such nonlinear waves in a corresponding four-dimensional phase space is mathematically challenging, and will be tackled using GSPT -- a dynamical systems approach to singularly perturbed ordinary differential equations introduced by Fenichel \cite{Fe79}. In its most common form, GSPT focuses on \emph{slow-fast} systems where the independent variable is considered to be time and the scale separation is induced in the equations by the presence of a small parameter $0 < \varepsilon \ll 1$. However, as in the case analysed here, the independent variable can be equally considered as a travelling-wave variable. Two equivalent formulations of the original problem (on the slow and fast scale) can be obtained by rescaling the independent variable by $\varepsilon$; in the singular limit $\varepsilon=0$, these two limiting systems are respectively called reduced problem and layer problem. Under appropriate assumptions and for sufficiently small values of $\varepsilon$, solutions can be obtained as perturbations of trajectories formed by piecing together solutions of the reduced and layer problems. We refer to \cite{JKK96, Kuehn_2015} for more background on GSPT and its many applications, which include in particular vegetation dynamics \cite{Sewalt2017,, Bastiaansen2019, Carter2018, Carter2024, grifo2025far, Iuorio2021}.

The structure of this paper is reported below and is designed to guide the reader through increasing levels of dynamical complexity.
In Section \ref{sec:model}, the mathematical model based on the hyperbolic extension of the Klausmeier model is presented. In order to account for the accumulation of harmful substances in the soil, autotoxic effects are also included in the model, in line with recent literature \cite{Marasco2014, Iuorio2021, grifo2025far}. 
Section \ref{sec:LSA} focuses on the definition of the pattern-forming region and on the identification of some classes of solutions, corresponding to structurally-different vegetation coherent structures, obtained as the distance from the instability threshold is varied.
In Section \ref{sec:closeto}, dynamics of vegetation patterns occurring close to onset are explored via LSA and WNA with the goal of discussing how inertia influences the selection of the wavelength and the migration speed together with the dynamical regime of the emerging patterns. 
Section \ref{sec:travpulse} is dedicated to far-from-equilibrium patterns, where GSPT is applied to study travelling pulses. The analysis is devoted to the rigorous construction of large-amplitude homoclinic solutions and to the characterisation of their multiscale structure, with particular emphasis on the role of inertia in shaping pulse speed and profile geometry. Theoretical predictions are complemented by numerical simulations illustrating how pulse amplitude, width, and migration velocity vary with the environmental parameters. Finally, in the Section \ref{sec:Conclusion}, our findings are synthesised and some perspectives on the ecological implications are proposed. 

\section{The model} \label{sec:model}
To address the study on the role of inertia in vegetation pattern dynamics, we consider a hyperbolic generalisation of the one-dimensional Klausmeier model for sloped semi-arid environments which also incorporates autotoxic effects \cite{Klausmeier1999, grifo2025far, Consolo2024}. According to the theory developed in \cite{BARBERA2015}, in dimensional form the model reads
\begin{equation} \label{eq:moddim}
\begin{aligned}
    \widetilde{U}_T-\nu \widetilde{U}_X &= p-l\widetilde{U}-r\widetilde{U} \widetilde{V}^2, \\
    \widetilde{V}_T+\widetilde{J}_X &= c \widetilde{U} \widetilde{V}^2-(d+\sigma \widetilde{S}) \widetilde{V}, \\
    \widetilde{S}_T &= q(d+\sigma \widetilde{S}) \widetilde{V} -(m+w p) \widetilde{S}, \\
    \widetilde{\tau}\widetilde{J}_T + \widetilde{V}_X &= -\frac{\widetilde{J}}{D_V} .
\end{aligned}
\end{equation}

The vector of state variables $\widetilde{\mathbf{W}}(X,T)=[\widetilde{U}(X,T),\widetilde{V}(X,T),\widetilde{S}(X,T),\widetilde{J}(X,T)]^T$ describes the spatio-temporal evolution of surface water, vegetation biomass, autotoxicity, and vegetation flux densities, respectively, at position $X\in\Omega\subset \mathbb{R}$ (positive $X$ direction being uphill) and time $T\in\mathbb{R}^+$. Regarding transport mechanisms, the advection term for surface water $-\nu \widetilde{U}_X$ mimics the downhill water flow with speed $\nu$, proportional to the slope of the terrain. Vegetation flux does not obey the classical Fick's law $\widetilde{J}=-D_V \widetilde{V}_X$ ($D_V$ being the diffusion coefficient), as usual in parabolic models. Instead, it satisfies a balance law where an inertial time $\widetilde{\tau}$ is considered to account for the non-instantaneous relation between the flux and the gradient of vegetation density \cite{BARBERA2015, Ruggeri2021, Grifo2025II, Consolo2025}. Of course, for vanishing inertia $\widetilde{\tau}\rightarrow0$, Fick's law (and, in turn, the original parabolic model) is recovered.
Moreover, diffusion of autotoxicity $\widetilde{S}$ in the soil is assumed to be negligible, in line with previous works \cite{Marasco2014, Iuorio2021, Consolo2024}. 
Regarding kinetic terms, surface water is fed by precipitation (with the parameter $p$ mimicking the mean annual rainfall value), whereas it is reduced due to evaporation (with rate $l$) and uptake by roots (with rate $r$). Vegetation biomass grows due to water availability (with rate $c$), but decreases according to a twofold mortality mechanism: intrinsic (with rate $d$) and toxicity-induced (with rate $\sigma$). Finally, four mechanisms involving toxicity are here considered: growth due to decomposition of biomass (with rate  $qd$) and interaction with biomass (with rate $q\sigma$), decay due to natural factors (with rate $m$), and washing-out effects driven by precipitation (with rate $wp$). 

Model \eqref{eq:moddim} can be recast in dimensionless form by rescaling the system's variables and parameters as follows:
\begin{gather*}
    X = \sqrt{\frac{D_V}{l}} x, \quad T = \frac{t}{l}, \quad \widetilde{U} = \frac{\sqrt{l r}}{c} U, \quad \widetilde{V} = \sqrt{\frac{l}{r}} V, \quad \widetilde{S} = \frac{l q}{m+w p} \sqrt{\frac{l}{r}} S, \quad \widetilde{J} = l \sqrt{\frac{D_V}{r}} J, \\
    \mathcal{A} = \frac{c p}{l \sqrt{l r}}, \quad \mathcal{B} = \frac{d}{l}, \quad \mathcal{D} = \frac{l}{m+w p}, \quad \mathcal{H} = \frac{q \sigma}{m + w p}\sqrt{\frac{l}{r}}, \quad \mathcal{V} = \frac{\nu}{\sqrt{D_V l}}, \quad \widetilde{\tau} = \frac{\tau\sqrt{r}}{D_V l \sqrt{l}},
\end{gather*}
which leads to 
\begin{equation} \label{eq:modadim}
\begin{aligned}
    U_t-\mathcal{V}U_x &= \mathcal{A}-U-U V^2 =: f(U,V), \\
    V_t+J_x &= U V^2-\mathcal{B}V-\mathcal{H}SV =: g(U,V,S), \\
    \mathcal{D} S_t &= \mathcal{B} V + \mathcal{H}S V-S =: h(V,S), \\
    \tau J_t +V_x &= -J, 
\end{aligned}
\end{equation}
or, in vector form,
\begin{equation}
\mathbf{W}_{t}+M\mathbf{W}_{x}=\mathbf{N}(\mathbf{W})  
\label{model_compact}
\end{equation}%
with%
\begin{equation}
\begin{tabular}{lll}
$\mathbf{W}=\left[ 
\begin{array}{c}
U \medskip \\ 
V \medskip \\ 
S \medskip \\ 
J%
\end{array}%
\right] ,$ & $M=\left[ 
\begin{tabular}{llll}
$-\mathcal{V}$ & $0$ & $0$ & $0$ \medskip \\ 
$0$ & $0$ & $0$ & $1$ \medskip \\ 
$0$ & $0$ & $0$ & $0$ \medskip \\ 
$0$ & $\frac{1}{\tau }$ & $0$ & $0$%
\end{tabular}%
\right] $, & $\mathbf{N}=\left[ 
\begin{array}{c}
f(U,V) \medskip \\ 
g(U,V,S) \medskip \\ 
h(V,S) \medskip \\ 
-\frac{1}{\tau }J%
\end{array}%
\right] $.%
\end{tabular}%
\label{vectors_model}
\end{equation}

The dimensionless version of the model thus consists of a reduced amount of parameters, $\mathcal{A}$, $\mathcal{B}$, $\mathcal{H}$, $\mathcal{D}$, $\mathcal{V}$ and $\tau$, which are proportional to rainfall, plant loss, plant's sensitivity to toxicity, decomposition rate of toxic compounds, slope of the terrain, and vegetation inertial time, respectively.

\section{Pattern-forming region}
\label{sec:LSA}
In this chapter, we define the region of the parameter plane $(\mathcal{B},\mathcal{A})$ in which stationary or migrating vegetation patterns may be observed. To this aim, let us first compute the spatially homogeneous steady-states of Eq.~\eqref{eq:modadim} and then carry out a linear stability analysis finalised at characterising the onset of diffusion-driving instabilities. These analyses will allow us to identify the model parameters for which pattern dynamics occur at, close to and far from the instability threshold. 

The spatially homogeneous steady-states %
\mbox{$\mathbf{W}^{\ast} = [U^\ast,V^\ast,S^\ast,J^\ast]^T$} admitted by Eq.~\eqref{eq:modadim} are given by $\mathbf{W}
_{D}^{\ast}$ and $\mathbf{W}_{R,L}^{\ast}$, representing respectively a desert state and two vegetated states, namely
\begin{equation}
\begin{aligned}
\mathbf{W}_{D}^{\ast} &= \left[ \mathcal{A},0,0,0\right]^T, \\ 
\mathbf{W}_{R,L}^{\ast} &= \left[ \frac{\mathcal{B}}{%
V_{R,L}\left( 1-\mathcal{H}V_{R,L}\right) }, V_{R,L}, \frac{%
\mathcal{B}V_{R,L}}{1-\mathcal{H}V_{R,L}},0\right]^T
\end{aligned}
\label{equilibria}
\end{equation}%
with
\begin{equation}
0<V_{L}=\frac{\mathcal{A}-\sqrt{\mathcal{A}^{2}-4\mathcal{B}\left( \mathcal{B%
}+\mathcal{A}\mathcal{H}\right) }}{2\left( \mathcal{B}+\mathcal{A}\mathcal{H}%
\right) }<\sqrt{1+\mathcal{H}^{2}}-\mathcal{H}<V_{R}=\frac{\mathcal{A}+\sqrt{%
\mathcal{A}^{2}-4\mathcal{B}\left( \mathcal{B}+\mathcal{A}\mathcal{H}\right) 
}}{2\left( \mathcal{B}+\mathcal{A}\mathcal{H}\right) }.%
\label{eq:VLVR}
\end{equation}%
The desert state $\mathbf{W}
_{D}^{\ast}$ exists for all values of $\mathcal{A}$, whereas the vegetated states $\mathbf{W}_{R,L}^{\ast}$ are admissible only for $\mathcal{A}>\mathcal{A}_{ex}$ where
\begin{equation}
\mathcal{A}_{ex}=2\mathcal{B}\left( \mathcal{H}+\sqrt{1+\mathcal{H}^{2}}%
\right).
\label{eq:ex}
\end{equation}
Moreover, in order for these steady-states to be ecologically feasible the condition $\mathcal{H}V_{R,L}<1$ must be fulfilled. For $\mathcal{A}=\mathcal{A}_{ex}$, the two homogeneous
vegetation steady-states coincide, leading to $\mathbf{W}_{R}^{\ast
}\equiv \mathbf{W}_{L}^{\ast}=\left[ \frac{\mathcal{B}\left( \left( 1+2%
\mathcal{H}^{2}\right) \sqrt{1+\mathcal{H}^{2}}+2\mathcal{H}\left( 1+%
\mathcal{H}^{2}\right) \right) }{1+\mathcal{H}^{2}},\sqrt{1+\mathcal{H}^{2}}-%
\mathcal{H},\frac{\mathcal{B}\sqrt{1+\mathcal{H}^{2}}}{1+\mathcal{H}^{2}}%
,0\right]^T$. 
To establish the local stability character associated with these steady-states, we
linearise Eq.~\eqref{model_compact} around the generic equilibrium $\mathbf{W}^{\ast}$ by looking for
solutions of the form 
\begin{equation}
\mathbf{W=W}^{\ast}+\widehat{\mathbf{W}}\exp \left( \omega t+\mathrm{i}%
kx\right),  \label{svi1}
\end{equation}%
where $\omega $ and $k$ are the growth factor and the real wavenumber,
respectively. This leads to 
\begin{equation}
\left( \omega {I}+\mathrm{i}k{M}-\mathcal{L}^{\ast}\right) 
\widehat{\mathbf{W}}=\mathbf{0},  \label{linear}
\end{equation}%
where ${I}$ is the 4 $\times$ 4 identity matrix, $\mathcal{L}^{\ast
}=\left( \nabla _{\mathbf{W}}\mathbf{N}\right)^{\ast}$, and the
asterisk denotes that the quantity is evaluated at $\mathbf{W}^{\ast}$. Meaningful solutions of \eqref{linear} are obtained if the growth factor $\omega $ satisfies the characteristic equation 
\begin{equation}
A_{0} \omega^{4}+A_{1}\omega^{3}+A_{2}\omega^{2}+A_{3}\omega +A_{4}=0,
\label{poliqua}
\end{equation}%
with 
\begin{equation}
\begin{array}{l}
A_{0}=\tau, \vspace{0.1cm}  \\
A_{1}=\alpha _{1}-\mathrm{i}k\mathcal{V\tau }, \vspace{0.1cm} \\ 
A_{2}=\alpha _{2}+k^{2}+\mathrm{i}k\mathcal{V}\left[ \left( g_{V}^{\ast
}+h_{S}^{\ast}\right) \tau -1\right], \vspace{0.1cm}\\ 
A_{3}=\alpha _{3}-k^{2}\left( f_{U}^{\ast}+h_{S}^{\ast}\right) +\mathrm{i}k%
\mathcal{V}\left[ \left( g_{V}^{\ast}+h_{S}^{\ast}-k^{2}\right) +\left(
g_{S}^{\ast}h_{V}^{\ast}-g_{V}^{\ast}h_{S}^{\ast}\right) \tau \right], \vspace{0.1cm}\\ 
A_{4}=\beta _{3}+f_{U}^{\ast}h_{S}^{\ast}k^{2}+\mathrm{i}k\mathcal{V}\left(
g_{S}^{\ast}h_{V}^{\ast}-g_{V}^{\ast}h_{S}^{\ast}+h_{S}^{\ast
}k^{2}\right), \vspace{0.1cm}\\ 
\alpha _{1}=1+\tau \beta _{1}, \vspace{0.1cm}\\ 
\alpha _{2}=\tau \beta _{2}+\beta _{1}, \vspace{0.1cm}\\ 
\alpha _{3}=\tau \beta _{3}+\beta _{2}, \vspace{0.1cm}\\ 
\beta _{1}=-\left( f_{U}^{\ast}+g_{V}^{\ast}+h_{S}^{\ast}\right), \vspace{0.1cm}\\ 
\beta _{2}=f_{U}^{\ast}g_{V}^{\ast}-f_{V}^{\ast}g_{U}^{\ast}+f_{U}^{\ast
}h_{S}^{\ast}+g_{V}^{\ast}h_{S}^{\ast}-g_{S}^{\ast}h_{V}^{\ast}, \vspace{0.1cm}\\ 
\beta _{3}=f_{U}^{\ast}g_{S}^{\ast}h_{V}^{\ast}-f_{U}^{\ast}g_{V}^{\ast
}h_{S}^{\ast}+f_{V}^{\ast}g_{U}^{\ast}h_{S}^{\ast}.%
\end{array}
\label{coef}
\end{equation}
As well known, $\mathbf{W}^{\ast}$ is linearly stable if all roots of \eqref{poliqua} exhibit negative real part for any wavenumber $k$. This analysis can be addressed via the
Routh-Hurwitz criterion as follows%
\begin{equation}
\begin{aligned}
\text{Re}\{\omega_i\} <0\quad \text{for  } i=1,\dots,4 \quad \Leftrightarrow \quad %
&A_{0}>0, \\
&A_{1}>0, \\
&A_{3}>0, \\
&A_{4}>0, \\
&A_{1}A_{2}A_{3}>{A_{0}}%
A_{3}^{2}+A_{1}^{2}A_{4}, \quad \forall k.  \label{routhhur}
\end{aligned}
\end{equation}
It can be easily checked that the characteristic equation \eqref{poliqua}, evaluated at the desert state $\mathbf{W}_{D}^{\ast}$, can be factorised as
\begin{equation}
\left( \omega \mathcal{+}1-\mathrm{i}k\mathcal{V}\right) \left( \omega 
\mathcal{+}\frac{1}{\mathcal{D}}\right) \left[ \omega^{2}+\left( \frac{1}{%
\tau }+\mathcal{B}\right) \omega +\frac{k^{2}+\mathcal{B}}{\tau }\right] =0,
\end{equation}%
so that $\mathbf{W}_{D}^{\ast}$ is always stable under both homogeneous and heterogeneous perturbations.

To inspect the stability character of the two vegetated states $\mathbf{W}_{R,L}^{\ast}$, let us distinguish the behaviour under homogeneous ($k=0$) and heterogeneous
($k\neq0$) perturbations. In the former case, the characteristic equation %
\eqref{poliqua} reduces to 
\begin{equation}
\left( \tau \omega +1\right) \left( \omega^{3}+\beta _{1}\omega^{2}+\beta
_{2}\omega +\beta _{3}\right) =0
\label{WLR_polchar}
\end{equation}%
and, being $\omega _{1}$ $=$ $-1/\tau$ always real and
negative, the stability character of these steady-states depends on the sign of the three roots of the cubic polynomial appearing in \eqref{WLR_polchar}. In this case, the Routh-Hurwitz criterion states that
\begin{equation}
\begin{aligned}
\text{Re}\{\omega_i\} <0\quad \text{for  } i=2,3,4  \quad \Leftrightarrow \quad & \beta
_{1}>0, \\
&\beta _{3}>0, \\
&\beta _{1}\beta _{2}-\beta _{3}>0, \quad \forall k.
\end{aligned}
\end{equation}%
Therefore, the vegetated state $\mathbf{W}_{L}^{\ast}$ is unstable because
the condition $\,\beta _{3}>0$ is always violated, whereas $\mathbf{W}_{R}^{\ast}$ can be linearly stable under homogeneous perturbations, depending on the model parameters.
These results suggest to exclude further investigations on $\mathbf{W}_{L}^{\ast}$ and to focus on possible destabilisation mechanisms about the steady-state $\mathbf{W}_{R}^{\ast}$. In particular, the interest is in deducing those conditions under which spatial disturbances acting on a stable spatially-uniform steady-state trigger the onset of Turing or wave instabilities responsible for the formation of a stationary or an oscillating periodic patterned state, respectively.

For a Turing instability to occur, it is necessary that the characteristic equation admits a null eigenvalue $\omega =0$ for a non-null wavenumber, which is equivalent to requiring in \eqref{poliqua} that $A_{4}=0$ for $k\not=0$. Unfortunately, the possibility to observe stationary Turing patterns is prevented since the condition $f_{V}^{\ast}g_{U}^{\ast
}h_{S}^{\ast}=0$ is never fulfilled. 

On the other hand, for a wave instability to occur, it is necessary to look for roots of \eqref{poliqua} with null real part for non-null values of the wavenumber. Such an analysis can be addressed by recasting the perturbation \eqref{svi1} in the form of a travelling wave, mimicking the motion of a migrating pattern, by adopting the ansatz $\omega =-\mathrm{i}\mathcal{C}k$, where $\mathcal{C}$ denotes the migration speed of the patterned solution. By substituting this ansatz into the characteristic equation \eqref{poliqua} and taking its derivative with
respect to $k$, we have that Eq.~\eqref{model_compact} undergoes a wave bifurcation with critical wavenumber
\begin{equation}
k_{w}^{2}=\frac{\mathcal{V}\left( g_{V}^{\ast}h_{S}^{\ast}-h_{V}^{\ast
}g_{S}^{\ast}\right) +\alpha _{3}\mathcal{C}}{\alpha _{1}\mathcal{C}^{3}+%
\mathcal{V}\left[ 1-\left( g_{V}^{\ast}+h_{S}^{\ast}\right) \tau \right] 
\mathcal{C}^{2}+\left( f_{U}^{\ast}+h_{S}^{\ast}\right) \mathcal{C}%
+h_{S}^{\ast}\mathcal{V}}  \label{eq:kc}
\end{equation}
whenever the following conditions hold
\begin{equation}
\begin{cases}
\mathcal{C}\left( \mathcal{C}+\mathcal{V}\right) \left( \tau \mathcal{C}%
^{2}-1\right) \left[ \alpha _{3}\mathcal{C}+\mathcal{V}\left( g_{V}^{\ast
}h_{S}^{\ast}-g_{S}^{\ast}h_{V}^{\ast}\right) \right]^{2}+\left[ \alpha
_{3}\mathtt{c+}\mathcal{V}\left( g_{V}^{\ast}h_{S}^{\ast}-g_{S}^{\ast
}h_{V}^{\ast}\right) \right]\times \\ 
\left\{ \mathcal{V}\mathcal{C}\left[ g_{V}^{\ast}+h_{S}^{\ast}-\tau \left(
g_{V}^{\ast}h_{S}^{\ast}-g_{S}^{\ast}h_{V}^{\ast}\right) \right] -\alpha
_{2}\mathcal{C}^{2}+f_{U}^{\ast}h_{S}^{\ast}\right\} \times \\ 
\left\{ \alpha _{1}\mathcal{C}^{3}+\mathcal{V}\mathcal{C}^{2}\left[ 1-\tau
\left( g_{V}^{\ast}+h_{S}^{\ast}\right) \right] +\left( f_{U}^{\ast
}+h_{S}^{\ast}\right) \mathcal{C}+h_{S}^{\ast}\mathcal{V}\right\} + \\ 
+\beta _{3}\left\{ \alpha _{1}\mathcal{C}^{3}+\mathcal{V}\mathcal{C}^{2}\left[
1-\tau \left( g_{V}^{\ast}+h_{S}^{\ast}\right) \right] +\left( f_{U}^{\ast
}+h_{S}^{\ast}\right) \mathcal{C}+h_{S}^{\ast}\mathcal{V}\right\}^{2}=0,
\bigskip \\ 
2\left\{ \left[ 3\alpha _{1}\mathcal{C}^{2}+2\mathcal{V}\mathtt{%
c-}2\mathcal{V}\mathcal{C}\left( g_{V}^{\ast}+h_{S}^{\ast}\right) \tau
+f_{U}^{\ast}+h_{S}^{\ast}\right] k_{w}^{2}-\alpha _{3}\right\} \times \\ 
\left\{ \alpha _{2}\mathcal{C}^{2}-\mathcal{V}\mathcal{C}\left[ g_{V}^{\ast
}+h_{S}^{\ast}-\tau \left( g_{V}^{\ast}h_{S}^{\ast}-h_{V}^{\ast
}g_{S}^{\ast}\right) \right] -f_{U}^{\ast}h_{S}^{\ast}-2\mathcal{C}\left( 
\mathcal{C}+\mathcal{V}\right) \left( \tau \mathcal{C}^{2}-1\right)
k_{w}^{2}\right\} +\medskip \\ 
+\{3k_{w}^{2}[ \alpha _{1}\mathcal{C}^{3}+\mathcal{V}\mathcal{C}%
^{2}\left[ 1-\left( g_{V}^{\ast}+h_{S}^{\ast}\right) \tau \right] +\left(
f_{U}^{\ast}+h_{S}^{\ast}\right) \mathcal{C}+h_{S}^{\ast}\mathcal{V}]
+\mathcal{V}(g_{S}^{\ast}h_{V}^{\ast}-g_{V}^{\ast} h_{S}^{\ast
}) - \mathcal{C} \alpha_3\} \times \\ 
\left\{ \left( 4\tau \mathcal{C}^{3}+3\tau \mathcal{V}\mathcal{C}^{2}-2\mathtt{%
c}-\mathcal{V}\right) k_{w}^{2}-2\alpha _{2}\mathcal{C}+\mathcal{V}\left[
g_{V}^{\ast}+h_{S}^{\ast}-\tau \left( g_{V}^{\ast}h_{S}^{\ast
}-g_{S}^{\ast}h_{V}^{\ast}\right) \right] \right\} =0.%
\end{cases}
\label{eq:Bc}
\end{equation}%

From now on, in order to investigate the behaviour of the ecosystem under variations of rainfall, we consider $\mathcal{A}$ as the main control parameter. Therefore, System \eqref{eq:Bc} defines implicitly the
migration speed $\mathcal{C}$ at onset and the critical rainfall value $\mathcal{A}_c$ at which wave instability occurs. 
Due to the nonlinear and nontrivial dependence of \eqref{eq:Bc} on
the model parameters, quantitative information on the region where migrating patterns can be observed is deduced via
numerical investigations, as reported in Fig.~\ref{firstfigure}.
Here, the upper curve denotes the wave instability locus $\mathcal{A}=\mathcal{A}_{c}$ obtained by solving numerically System \eqref{eq:kc}-\eqref{eq:Bc} for the parameter set reported in the caption, whereas the straight line represents the existence threshold $\mathcal{A}=\mathcal{A}_{ex}$ for the state $\mathbf{W}_{R}^{\ast}$ given in \eqref{eq:ex}. These two loci subdivide the $\left(\mathcal{B},\mathcal{A}\right)$ parameter plane into three different regions: (i) for $\mathcal{A}<\mathcal{A}_{ex}$, only the desert state exists and is stable; (ii) for $\mathcal{A}>\mathcal{A}_{c}$, stable homogeneous vegetation cover ($\mathbf{W}_{R}^{\ast}$) and bare ground ($\mathbf{W}_{D}^{\ast}$) states coexist; (iii) for $\mathcal{A}_{ex}<\mathcal{A}<\mathcal{A}_{c}$, oscillatory periodic (migrating banded) patterns may be observed. 
\begin{figure}[t!]
	\centering
	\includegraphics[width=0.5\textwidth]{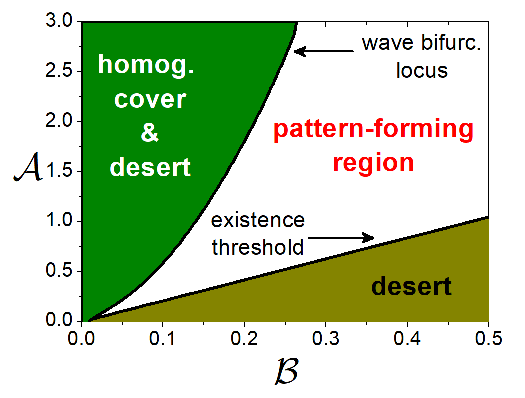}
	\caption{Subdivision of the $\left(\mathcal{B},\mathcal{A}\right)$-plane into three different zones according to the location of the existence threshold $\mathcal{A}=\mathcal{A}_{ex}$ and the wave bifurcation locus $\mathcal{A}=\mathcal{A}_{c}$ obtained via numerical integration of System \eqref{eq:kc}-\eqref{eq:Bc}. Parameter set: $\tau=1$, $\mathcal{H}=0.05$, $\mathcal{V} = 182.5$ and $\mathcal{D} = 4.5$. Note that a different choice of parameter values does not affect qualitatively the abovementioned subdivision of the parameter plane.}	
	\label{firstfigure}	
\end{figure}

This preliminary result allows us to explore, qualitatively, the pattern dynamics occurring into the wave instability region. To this aim, the governing system \eqref{model_compact}-\eqref{vectors_model} is integrated numerically through COMSOL Multiphysics \cite{Comsol} by using a computational domain $x\in[0,100]$ together with periodic boundary conditions. Then, starting from a configuration $(\mathcal{B}_c,\mathcal{A}_c)$ lying in the wave locus, two ecological scenarios reproducing the worsening of environmental conditions are considered: (i) increasing aridity, where the rainfall value is progressively decreased from $\mathcal{A}_c$, keeping the plant loss fixed and (ii) increasing plant mortality, where the plant loss is progressively increased from $\mathcal{B}_c$, keeping the rainfall constant. The above ramps occur over a time window $t\in[0,1000]$, and the initial condition is represented by a small random perturbation of the homogeneous steady-state $\mathbf{W}_{R}^{\ast}$. Results of the former investigation are reported in Fig.~\ref{figure0}(a-d), those of the latter in Fig.~\ref{figure0}(e-h). 
In both cases, the emerging migrating patterns exhibit qualitatively similar behaviours. In particular, very close to onset, migrating patterns are in the form of small-amplitude periodic oscillations around the equilibrium vegetated state (panels (b) and (e)). When the distance from the threshold increases, the amplitude of the oscillation increases too, the minimum of vegetation density approaches the bare ground, and the overall pattern shape deviates from that of a sinusoid (panels (c) and (f)). Very far from the threshold, the patterned solution takes the form of large-amplitude travelling periodic pulses that originate from the desert state and whose profile exhibits different spatial scales (panels (d) and (g)).
Also, the worsening of environmental conditions yields an increase of the migration speed but a reduction of pattern wavelength (panels a and h).
\begin{figure}[p]
	\centering
	\includegraphics[width=0.75\textwidth]{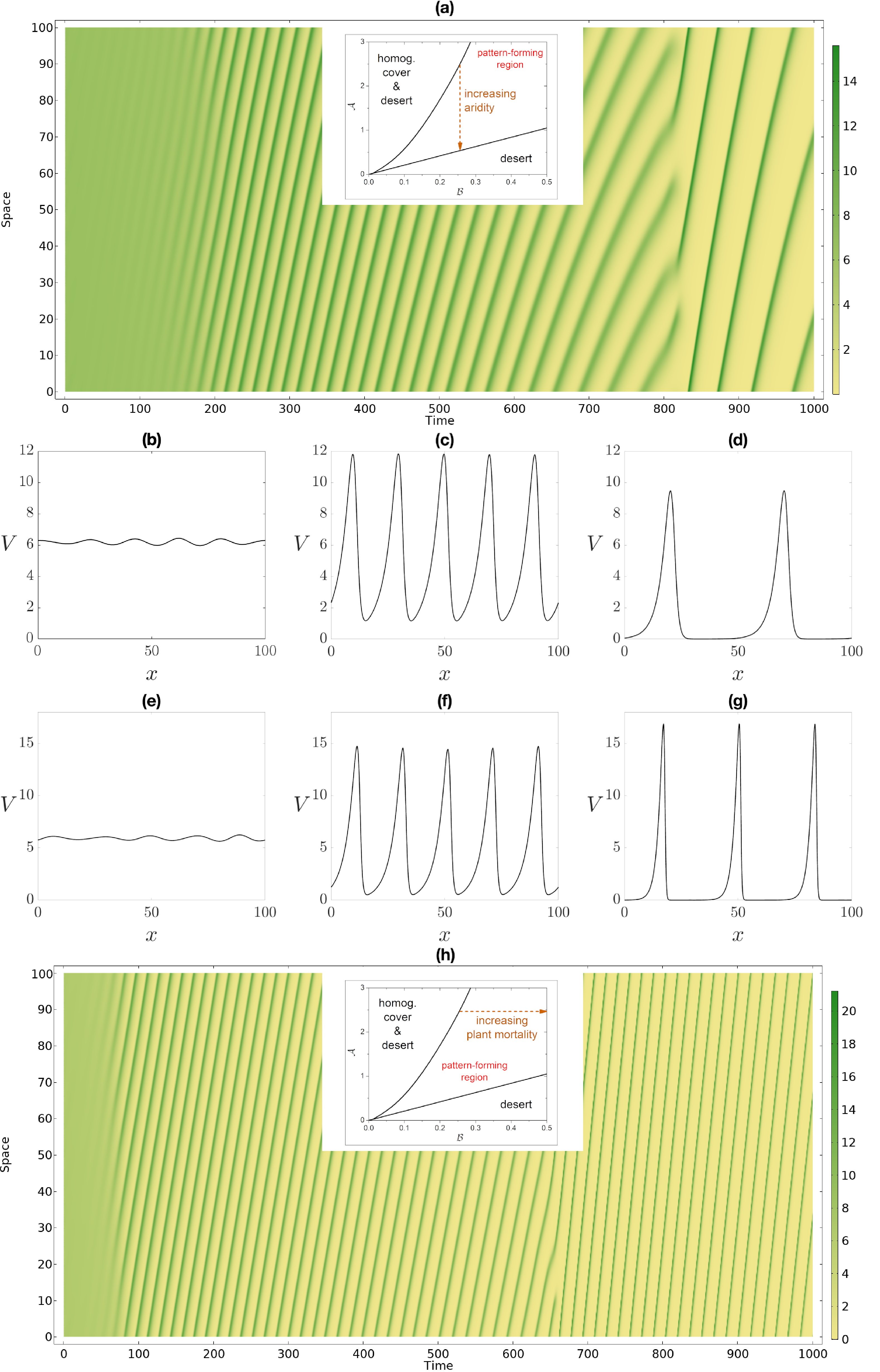}
	\caption{Vegetation pattern dynamics observed into the wave instability region under worsening environmental conditions for increasing aridity (panels a-d) and increasing plant mortality (panels e-h). Panels (a) and (h) show the spatio-temporal evolution of the patterned solution obtained by sweeping the rainfall or the plant loss, respectively, starting from the critical value $(\mathcal{B}_c,\mathcal{A}_c)$. Panels (b,e), (c,f) and (d,g) depict the spatial profiles of such solutions very close to the threshold, in the middle of the pattern-forming region and very far from the instability threshold, respectively. Parameter set: $\tau=1$, $\mathcal{H}=0.05$, $\mathcal{V} = 182.5$ and $\mathcal{D} = 4.5$ (same as in Fig.~\ref{firstfigure}).}	
	\label{figure0}	
\end{figure}

The main aim of the subsequent analysis is to characterise in detail the two most relevant class of solutions depicted in Fig.~\ref{figure0} observed close to and far from onset: small-amplitude oscillating periodic patterns and travelling periodic patterns. In particular, we focus on inspecting how inertia affects their key features. 

\section{Vegetation patterns close to onset} \label{sec:closeto}
\subsection{Characterisation of the wave bifurcation locus}
In this section, we exploit numerical investigations to elucidate the functional dependence of the wave bifurcation locus on the model parameters and, in turn, to gain deeper insights into the corresponding ecological implications. Results are shown in Fig.~\ref{figure1}. Here, the solid black line represents the existence locus $\mathcal{A}=\mathcal{A}_{ex}$ given in \eqref{eq:ex}. 
Wave bifurcation loci are computed by solving System \eqref{eq:kc}-\eqref{eq:Bc} for different values of $\mathcal{H}$ ($0$ and $0.05$) and $\tau$, namely $\tau\leq0.1$ (solid blue line), $\tau=1$ (dash-dotted red line) and $\tau=3$ (dashed green line). Inspection of these results reveals that increasing sensitivity to autotoxicity and inertia favors the enlargment of the pattern-forming region, as they both act as destabilising mechanisms \cite{Consolo2022PRE, Consolo2024, grifo2025far}. In particular, it should be noticed that the parabolic regime (negligible inertial effects) extends up to $\tau\leq0.1$. Significant inertial effects become tangible when the inertial time is increased by about one order of magnitude, leading to a relevant expansion of the pattern-forming region, especially for large values of rainfall.
\begin{figure}[t!]
	\centering
	\includegraphics[width=1\textwidth]{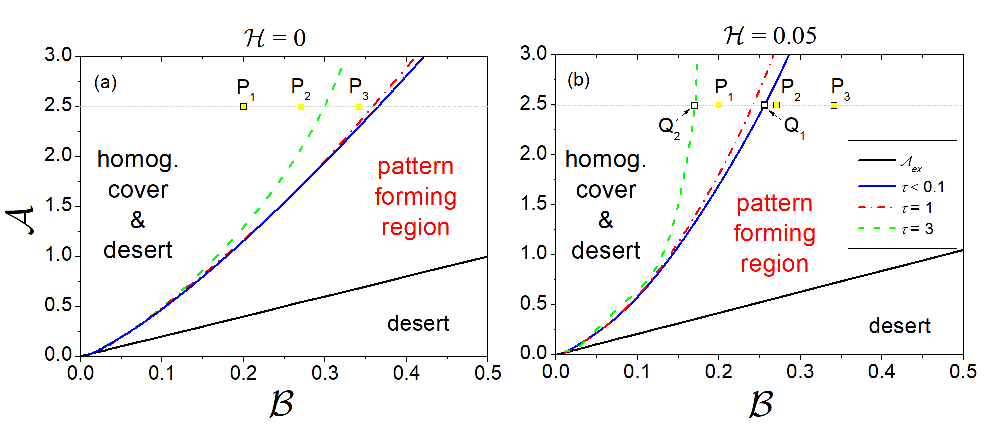}
	\caption{Wave bifurcation locus in the $\left(\mathcal{B},\mathcal{A}\right)$-plane for two different values of the toxicity strength: (a) $\mathcal{H}=0$, (b) $\mathcal{H}=0.05$. The loci are computed for different values of inertial time: $\tau\leq0.1$ (solid blue line), $\tau=1$ (dash-dotted red line) and $\tau=3$ (dashed green line). The solid black line denotes the existence locus $\mathcal{A}=\mathcal{A}_{ex}$.}	
	\label{figure1}	
\end{figure}

A validation of the previous predictions is achieved by means of numerical integration of the governing system \eqref{model_compact}-\eqref{vectors_model} through COMSOL Multiphysics \cite{Comsol}. Numerical simulations are carried out by using a computational domain $x\in[0,100]$ over a time window $t\in[0,200]$, periodic boundary conditions, and considering a small random perturbation of the homogeneous steady-state $\mathbf{W}_{R}^{\ast}$ as initial condition. In Fig.~\ref{figure3} we report the solutions associated to the points $\mathrm{P}_1$ (first column), $\mathrm{P}_2$ (second column) and $\mathrm{P}_3$ (third column), respectively, and obtained for different values of $\mathcal{H}$ and $\tau$. Results again reveal an excellent agreement with the above findings. Indeed, in all those cases reported in Fig.~\ref{figure1} where the point falls within the pattern-forming region, the numerical integration of the system gives rise to oscillatory patterned dynamics, namely vegetation bands migrating uphill. On the contrary, if the point lies beyond the wave bifurcation locus, the system converges towards the uniformly vegetated state. Of course, in the same figure, it is possible to appreciate patterned states with variable entities of migration speed and wavelengths, as they reproduce dynamics occurring at different distances from the instability threshold and at different rainfall values. 
\begin{figure}[p!]
	\centering
	\includegraphics[width=1\textwidth]{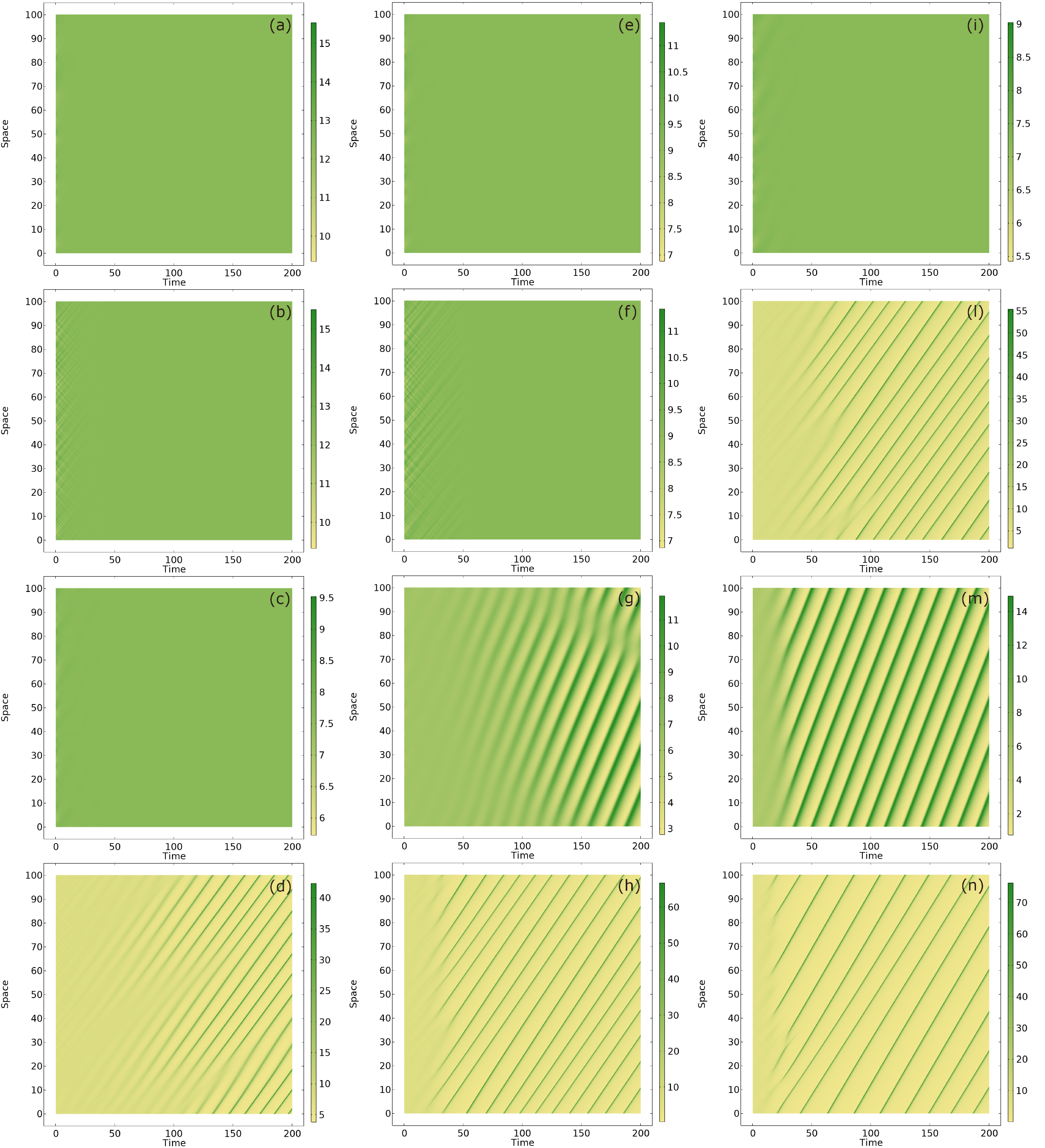}
	\caption{Vegetation patterns obtained by integrating numerically the governing System \eqref{model_compact}-\eqref{vectors_model} by using the parameter set $(\mathcal{B},\mathcal{A})$ corresponding to points $\mathrm{P}_1$ (first column), $\mathrm{P}_2$ (second column) and $\mathrm{P}_3$ (last column) depicted in Fig.~\ref{figure1}. The panels represent spatio-temporal evolutions attained for different values of autotoxicity strength and inertial time: $\mathcal{H}=0$ and $\tau=0.1$ (first row), $\mathcal{H}=0$ and $\tau=3$ (second row), $\mathcal{H}=0.05$ and $\tau=0.1$ (third row), and $\mathcal{H}=0.05$ and $\tau=3$ (last row).}	
	\label{figure3}	
\end{figure}

We now turn our attention to the dependence of the migration speed on the main model parameters and, in particular, on inertia. 
To this aim, we construct a density plot in the $\left(\mathcal{H},\mathcal{A}\right)$-plane for the migration speed $\mathcal{C}$ at the onset of wave instability. This result is obtained by solving numerically System \eqref{eq:kc}-\eqref{eq:Bc} for two values of inertial time: $\tau=0.1$ (see Fig.~\ref{figure6}(a)) and $\tau=3$ (see Fig.~\ref{figure6}(b)). These pictures provide several information on the migration speed of banded vegetation: (i) it increases with the increase of rainfall; (ii) for very small values of rainfall (about $\mathcal{A}<0.2$), it is almost independent of autotoxicity and gets quite small values ($\mathcal{C}\leq0.2$); (iii) for larger values of rainfall ($\mathcal{A}>0.2$), it becomes rainfall-dependent and the combination between the negative effect induced by autotoxicity (which speeds up the erosion process of the bottom part of the tiger bush) and the larger availability of water (giving rise to higher accumulation of soil moisture in the top part of vegetation band) favors a faster uphill migration of the overall pattern \cite{Tongway2001, Sherratt2012, SHERRATT2013, EIGENTLER2020b, Consolo2024II, EIGENTLER2024}; (iv) the presence of inertia tends to mitigate the role of autotoxicity, as weak variations of the migration speed are observed with variable toxicity strength; (v) the presence of inertia also slows down the growth of new vegetation above the upper edge of the tiger bush, so leading to overall smaller migration speeds (about 30$\%$ on average).   
\begin{figure}[t!]
	\centering
	\includegraphics[width=1\textwidth]{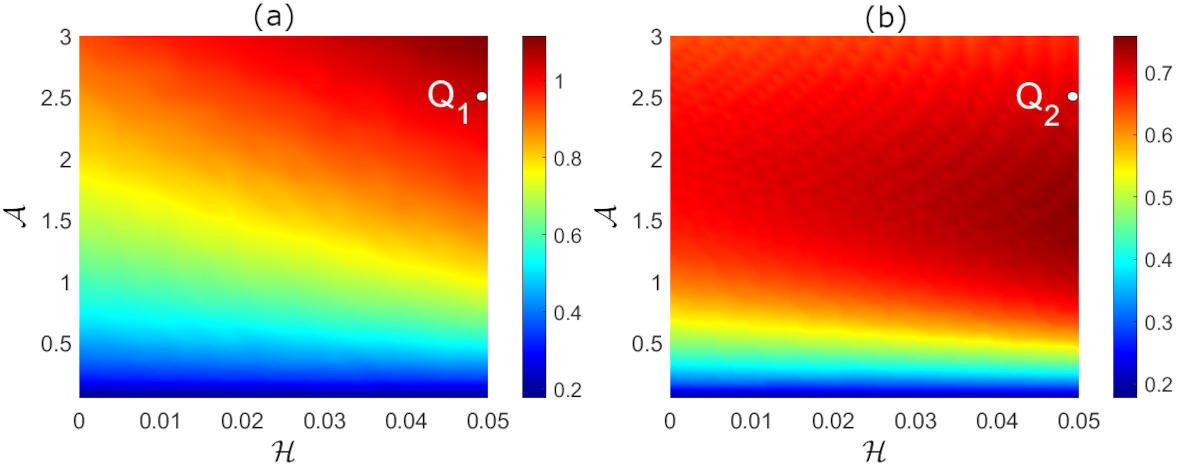}
	\caption{Migration speed $\mathcal{C}$ in the $\left(\mathcal{H},\mathcal{A}\right)$-plane at onset of wave instability obtained by solving System \eqref{eq:kc}-\eqref{eq:Bc} for (a) $\tau=0.1$ and (b) $\tau=3$.}	
	\label{figure6}	
\end{figure}

To verify whether these predictions on the migration speed deduced from LSA hold true, we again perform numerical simulations to integrate the governing System \eqref{model_compact}-\eqref{vectors_model} with a computational domain $x\in[0,x_D]=[0,100]$, a time window $t\in[0,t_{end}]=[0,500]$, periodic boundary conditions, and a small random perturbation of the homogeneous steady-state $\mathbf{W}_{R}^{\ast}$ as initial condition. For simplicity, we consider the parameter set $\left(\mathcal{H},\mathcal{A}\right)=(0.05, 2.5)$ and two different inertial times, $\tau=0.1$ and $\tau=3$. This setup is representative of the pattern dynamics observed at the points labeled as $\mathrm{Q}_1$ and $\mathrm{Q}_2$ in Fig.~\ref{figure6}, respectively. It is worth noticing that, since the migration speeds there reported are evaluated at onset through \eqref{eq:kc}-\eqref{eq:Bc}, the change of inertial time yields a shift of the wave locus, as previously described (see Fig.~\ref{figure1}). As a consequence of that, the scenario represented by the points $\mathrm{Q}_1$ and $\mathrm{Q}_2$ denotes vegetation dynamics taking place for different values of inertial time and plant loss, keeping rainfall and sensitivity to autotoxicity fixed, as also appears clear from Fig.~\ref{figure1}(b). For computational reasons, the values of plant loss $\mathcal{B}$ are set, in both cases, at a distance $0.1$ from the respective thresholds.
Results are depicted in Fig.~\ref{figure7} (panels (a) and (d)), where we also report the Fast Fourier Transforms (FFTs) of the vegetation density $V(x,t)$ evaluated either at a fixed location ($x=x_D/2$) (panels (b) and (e)) or at the final simulation time ($t=t_{end}$) (panels (c) and (f)). From the estimation of the location of the dominant peak in these graph, it is possible to estimate numerically the migration speed as $\mathcal{C}^{num}=\omega_0/k_0$ which is then compared with the theoretical one $\mathcal{C}^{th}$ reported in Fig.~\ref{figure6}. The analysis reveals a satisfying agreement in both cases: for $\tau=0.1$ we get $\mathcal{C}^{th}=1.04$ and $\mathcal{C}^{num}=1.03$ while, for $\tau=3$ we get $\mathcal{C}^{th}=0.71$ and $\mathcal{C}^{num}=0.73$. This also confirms that the reduction of migration speed is not only caused by inertia but there is an additional contribution due to plant loss (in line with observations reported in Fig.~\ref{figure0}).
\begin{figure}[tt!]
	\centering
	\includegraphics[width=1\textwidth]{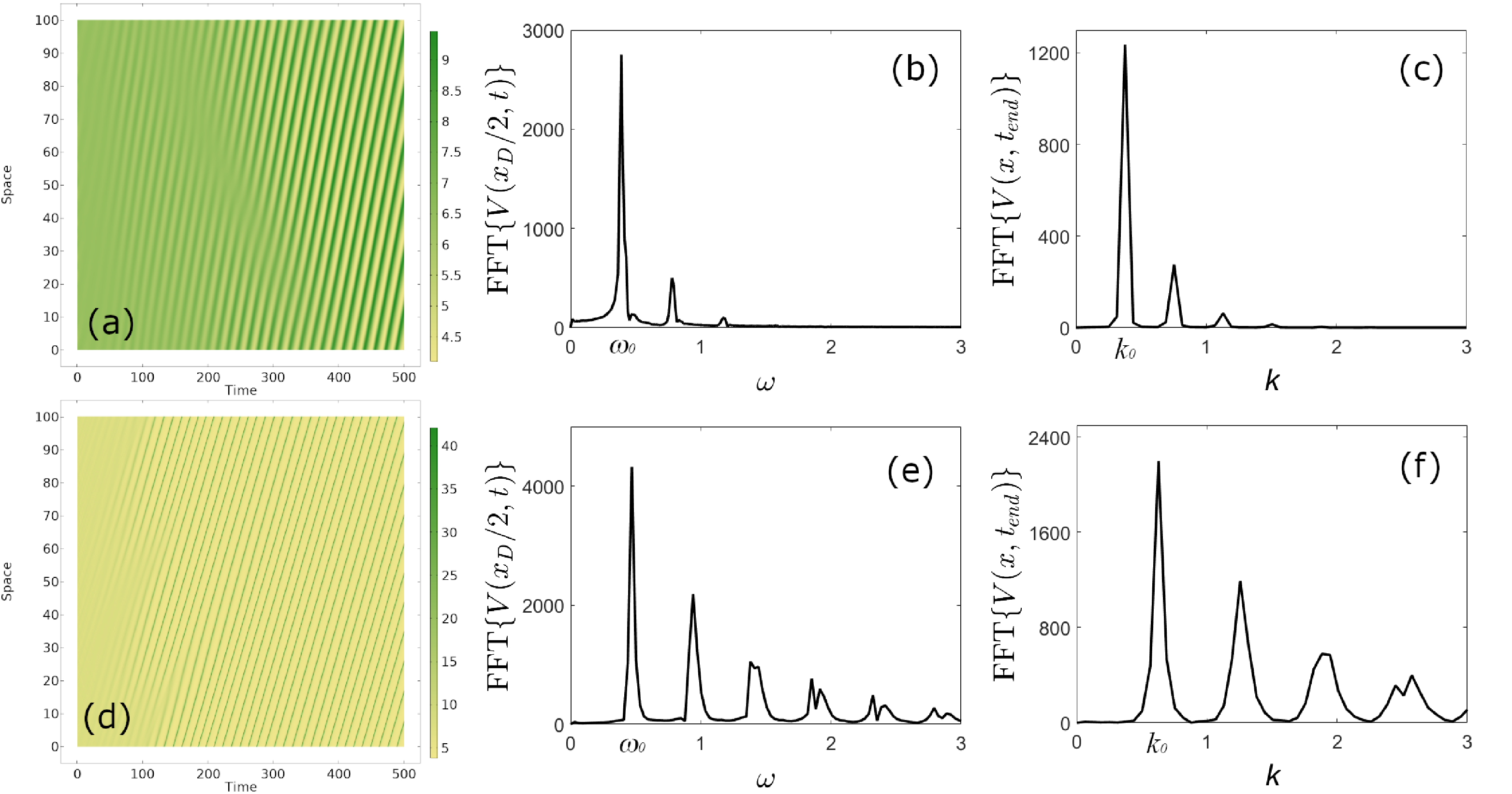}
	\caption{(a,d) Spatio-temporal vegetation patterns obtained by using the parameter set corresponding to points $\mathrm{Q}_1$ (a) and $\mathrm{Q}_2$ (d) depicted in Fig.~\ref{figure6} and in Fig.~\ref{figure1}(b). (b,e) FFTs of vegetation density evaluated at half of the computational domain FFT$\{V\left(x_D,t\right)\}$. (c,f) FFTs of vegetation density evaluated at the final simulation time FFT$\{V\left(x,t_{end}\right)\}$.}	
	\label{figure7}	
\end{figure}

\subsection{Evolution of pattern amplitude}
\label{sec:WNA}
To characterise the time evolution of the pattern amplitude near the instability onset, we perform a multiple-scale WNA \cite{Wollkind1994, Consolo2024, Giunta2021} around the steady-state $\mathbf{W}_R^{\ast}$ close to the critical value of the main control parameter $\mathcal{A}_c$ defined implicitly in \eqref{eq:kc}-\eqref{eq:Bc}.

First, we expand $\mathbf{N}\left( \mathbf{W}\right) $ around $\mathbf{W}_R^{\ast}$ and separate the linear and nonlinear parts by distinguishing
quadratic and cubic terms, thus recasting Eq.~\eqref{model_compact}-\eqref{vectors_model} into the following form
\begin{equation}
\overline{\mathbf{W}}_{t}=\mathcal{L}\overline{\mathbf{W}}+\frac{1}{2}%
\mathcal{Q}\left( \overline{\mathbf{W}},\overline{\mathbf{W}}\right) +\frac{1%
}{6}\mathcal{R}\left( \overline{\mathbf{W}},\overline{\mathbf{W}},\overline{%
\mathbf{W}}\right) ,  \label{eq:model_oper}
\end{equation}%
where $\overline{\mathbf{W}}=\mathbf{W}-\mathbf{W}_R^{\ast}$, $\mathcal{L}=%
\mathcal{L}^{\ast}-{M}\partial _{x}$, and the operators $\mathcal{Q}$
and $\mathcal{R}$ -- acting on the couple $\left( \mathbf{x},\mathbf{y}\right) $
and on the triplet $\left( \mathbf{x},\mathbf{y},\mathbf{z}\right) $,
respectively -- are expressed as%
\begin{eqnarray}
\mathcal{Q}\left( \mathbf{x},\mathbf{y}\right) &=&\left( \mathbf{x\cdot
\nabla }\right) \left( \mathbf{y\cdot \nabla }\right) \mathbf{N} \\
\mathcal{R}\left( \mathbf{x},\mathbf{y,z}\right) &=&\left( \mathbf{x\cdot
\nabla }\right) \left( \mathbf{y\cdot \nabla }\right) \left( \mathbf{%
z\cdot \nabla }\right) \mathbf{N}
\end{eqnarray}%
being, for a generic vector $\mathbf{w}$,
\[
\mathbf{w}\cdot \nabla =w_{1}%
\frac{\partial }{\partial U}+w_{2}\frac{\partial }{\partial V}+w_{3}\frac{%
\partial }{\partial S}+w_{4}\frac{%
\partial }{\partial J}.
\]

Then, we expand the field $\overline{\mathbf{W}}$ and the bifurcation
parameter $\mathcal{A}$ in terms of a positive small parameter $%
0<\epsilon \ll 1$, measuring the relative distance from the threshold, and account for a time hierarchy as follows:

\begin{equation}
\begin{array}{l}
\overline{\mathbf{W}}=\epsilon \overline{\mathbf{W}}_{1}+\epsilon^{2}\overline{\mathbf{W}}%
_{2}+\epsilon^{3}\overline{\mathbf{W}}_{3}+O\left( \epsilon^{4}\right) ,\medskip
\\ 
\mathcal{A}=\mathcal{A}_{c}+\epsilon \mathcal{A}_{1}+\epsilon^{2}%
\mathcal{A}_{2}+\epsilon^{3}\mathcal{A}_{3}+O\left( \epsilon
^{4}\right) ,\medskip \\ 
\frac{\partial }{\partial t}\rightarrow \frac{\partial }{\partial t}%
+\epsilon \frac{\partial }{\partial T_{1}}+\epsilon^{2}\frac{\partial 
}{\partial T_{2}}+\epsilon^{3}\frac{\partial }{\partial T_{3}}+O\left(
\epsilon^{4}\right) .%
\end{array}
\label{eq:sviluppo}
\end{equation}%
Furthermore, taking into account Eq.~\eqref{eq:sviluppo}, the linear, quadratic
and cubic operators $\mathcal{L}$, $\mathcal{Q}$, and $\mathcal{R}$ can be expanded as
\begin{equation}
\begin{array}{l}
\mathcal{L}=\mathcal{L}_{c}+\epsilon \mathcal{A}_{1}\mathcal{L}_{\mathcal{%
A}}+\epsilon^{2}\left( \mathcal{A}_{2}\mathcal{L}_{\mathcal{A}}+\frac{1}{%
2}\mathcal{A}_{1}^{2}\mathcal{L}_{\mathcal{A}\mathcal{A}}\right)
+\epsilon^{3}\left( \mathcal{A}_{3}\mathcal{L}_{\mathcal{A}}+\mathcal{A}%
_{1}\mathcal{A}_{2}\mathcal{L}_{\mathcal{A}\mathcal{A}}+\frac{1}{6}\mathcal{A%
}_{1}^{3}\mathcal{L}_{\mathcal{A}\mathcal{A}\mathcal{A}}\right) +O\left(
\epsilon^{4}\right) ,\medskip \\ 
\mathcal{Q}\left( \overline{\mathbf{W}},\overline{\mathbf{W}}\right)
=\epsilon^{2}\mathcal{Q}\left( \overline{\mathbf{W}}_{1},\overline{\mathbf{W}}_{1}\right)
+\epsilon^{3}\left[ 2\mathcal{Q}\left( \overline{\mathbf{W}}_{1},\overline{\mathbf{W}}%
_{2}\right) +\mathcal{A}_{1}\mathcal{Q}_{\mathcal{A}}\left( \overline{\mathbf{W}}_{1},%
\overline{\mathbf{W}}_{1}\right) \right] +O\left( \epsilon^{4}\right) , \medskip \\ 
\mathcal{R}\left( \overline{\mathbf{W}},\overline{\mathbf{W}},\overline{%
\mathbf{W}}\right) =\epsilon^{3}\mathcal{R}\left( \overline{\mathbf{W}}_{1},\mathbf{%
W}_{1},\overline{\mathbf{W}}_{1}\right) +O\left( \epsilon^{4}\right) ,%
\end{array}
\label{eq:exp_op}
\end{equation}%
where $\mathcal{L}_{c}=\mathcal{L}^{\ast}\vert_{\mathcal{A}=
\mathcal{A}_{c}}-%
{M}\partial _{x}$ and all the other operators here appearing are evaluated at $\mathcal{A}=\mathcal{A}_{c}$ as well. \\
\noindent
Substituting the expansions \eqref{eq:sviluppo}-\eqref{eq:exp_op} into %
Eq.~\eqref{eq:model_oper} and collecting terms of the same order of $%
\epsilon $, the following set of PDEs is obtained%
\begin{equation}
\begin{array}{ll}
\text{at order}\hspace{0.1cm}\epsilon: & \frac{\partial \overline{\mathbf{W}}_{1}}{%
\partial t}=\mathcal{L}_{c}\overline{\mathbf{W}}_{1},\medskip \\ 
\text{at order}\hspace{0.1cm}\epsilon^{2}: & \frac{\partial \overline{\mathbf{W}}_{2}}{%
\partial t}=\mathcal{L}_{c}\overline{\mathbf{W}}_{2}-\mathbf{F},\medskip \\ 
\text{at order}\hspace{0.1cm}\epsilon^{3}: & \frac{\partial \overline{\mathbf{W}}_{3}}{%
\partial t}=\mathcal{L}_{c}\overline{\mathbf{W}}_{3}-\mathbf{G},\medskip%
\end{array}
\label{ordini}
\end{equation}
where
\begin{equation}
\begin{array}{l}
\mathbf{F}=-\frac{1}{2}\mathcal{Q}\left( \overline{\mathbf{W}}_{1},\overline{\mathbf{W}}_{1}\right)
-\mathcal{A}_{1}\mathcal{L}_{\mathcal{A}}\overline{\mathbf{W}}_{1}+\frac{\partial 
\overline{\mathbf{W}}_{1}}{\partial T_{1}},\medskip \\ 
\mathbf{G}=-\mathcal{Q}\left( \overline{\mathbf{W}}_{1},\overline{\mathbf{W}}_{2}\right) -\frac{1}{6%
}\mathcal{R}\left( \overline{\mathbf{W}}_{1},\overline{\mathbf{W}}_{1},\overline{\mathbf{W}}_{1}\right) -%
\mathcal{A}_{1}\mathcal{L}_{\mathcal{A}}\overline{\mathbf{W}}_{2}-\mathcal{A}_{2}%
\mathcal{L}_{\mathcal{A}}\overline{\mathbf{W}}_{1}-\frac{1}{2}\mathcal{A}_{1}^{2}%
\mathcal{L}_{\mathcal{A}\mathcal{A}}\overline{\mathbf{W}}_{1}+\medskip \\ 
\hspace{1cm}-\frac{1}{2}\mathcal{A}_{1}\mathcal{Q}_{\mathcal{A}}\left( 
\overline{\mathbf{W}}_{1},\overline{\mathbf{W}}_{1}\right) +\frac{\partial \overline{\mathbf{W}}_{2}}{%
\partial T_{1}}+\frac{\partial \overline{\mathbf{W}}_{1}}{\partial T_{2}}.%
\end{array}
\label{eq:ordini_sor}
\end{equation}

To address the study of patterned solutions in the case of wave instability,
let us consider a travelling wave variable $\xi =x-\mathcal{C}t$, so that the PDEs
system \eqref{ordini} can be recast into the following set of ODEs

\begin{subequations}
\begin{alignat}{2}
&\text{at order}\hspace{0.1cm}\epsilon: \quad && \widehat{\mathcal{L}}\overline{\mathbf{W}}%
_{1}=\mathbf{0}, \label{ODE_ordini1} \medskip \\ 
&\text{at order}\hspace{0.1cm}\epsilon^{2}: \quad && \widehat{\mathcal{L}}\overline{\mathbf{W}}%
_{2}=\mathbf{F}, \label{ODE_ordini2} \medskip \\ 
&\text{at order}\hspace{0.1cm}\epsilon^{3}: \quad && \widehat{\mathcal{L}}\overline{\mathbf{W}}%
_{3}=\mathbf{G}, \label{ODE_ordini3}
\medskip%
\end{alignat}
\label{ODE_ordini}
\end{subequations}
being $\widehat{\mathcal{L}}=\mathcal{L}^{\ast}\vert_{\mathcal{A}=
\mathcal{A}_{c}}+\left( \mathcal{C}I-{M}\right) \partial _{\xi }$. 
At first order, the solution reads
\begin{equation}
\overline{\mathbf{W}}_{1}=\frac{1}{2}\left\{ \Omega \,\boldsymbol{\varphi }\,e^{\mathrm{i}%
\hat{k}\xi }+\overline{\Omega }\,\overline{\boldsymbol{\varphi }}\,e^{-\text{%
i}\hat{k}\xi }\right\} ,  \label{eq:sol_wave_ordine1}
\end{equation}
being $\overline{\hspace{0.1cm}\cdot \hspace{0.1cm}}$ the complex conjugate
and $\hat{k}$ the first unstable excited mode.
From the eigenvalue problem $\widehat{\mathcal{L}}\boldsymbol{\varphi }=\mathbf{0}$, we get $\hat{k}=k_w$ given in \eqref{eq:kc} and the vector 
\mbox{$\boldsymbol{\varphi }\in \text{Ker}\left\{
\mathcal{L}^\ast\vert_{\mathcal{A}=
\mathcal{A}_{c}} + \mathrm{i}\hat{k}\left(\mathcal{C}I-M\right)\right\}$} appearing in \eqref{eq:sol_wave_ordine1} is expressed as
\begin{equation}
\boldsymbol{\varphi }=\left[ 
\begin{array}{c}
1 \\ 
r_{2}+\mathrm{i}\hat{r}_{2} \\ 
r_{3}+\mathrm{i}\hat{r}_{3} \\ 
r_{4}+\mathrm{i}\hat{r}_{4}%
\end{array}%
\right] ,  \label{eq:autovettore_wave}
\end{equation}%
with 
\begin{equation}
\begin{alignedat}{2}
r_{2}   &= -\frac{f_{U}^{\ast}}{f_{V}^{\ast}},
&\qquad \hat{r}_{2} &= -\frac{(\mathcal{C}+\mathcal{V}) k_{w}}{f_{V}^{\ast}}, \\[0.7ex]
r_{3}   &= -\frac{h_{V}^{\ast}\left( h_{S}^{\ast}r_{2}+\mathcal{C}k_{w}\hat{r}_{2}\right)}
              {{h_{S}^{\ast}}^{2}+\mathcal{C}^{2}k_{w}^{2}},
&\qquad \hat{r}_{3} &= \frac{h_{V}^{\ast}\left( \mathcal{C}k_{w}r_{2}-h_{S}^{\ast}\hat{r}_{2}\right)}
              {{h_{S}^{\ast}}^{2}+\mathcal{C}^{2}k_{w}^{2}}, \\[0.7ex]
r_{4}   &= \frac{k_{w}\left( \mathcal{C}k_{w}\tau r_{2}+\hat{r}_{2}\right)}
              {\mathcal{C}^{2}k_{w}^{2}\tau^{2}+1},
&\qquad \hat{r}_{4} &= \frac{k_{w}\left( \mathcal{C}k_{w}\tau \hat{r}_{2}-r_{2}\right)}
              {\mathcal{C}^{2}k_{w}^{2}\tau^{2}+1}.
\end{alignedat}
\end{equation}
Therefore, at the second order, the non-homogeneous equation \eqref{ODE_ordini2} reduces to%
\begin{equation}
\begin{array}{l}
\widehat{\mathcal{L}}\mathbf{W}_{2}=-\frac{1}{8}\Omega^{2}\mathcal{Q}\left( 
\boldsymbol{\varphi },\boldsymbol{\varphi }\right) e^{2\mathrm{i}k_{w}\xi }-%
\frac{1}{4}|\Omega |^{2}\mathcal{Q}\left( \boldsymbol{\varphi },\overline{%
\boldsymbol{\varphi }}\right) -\frac{1}{8}\overline{\Omega }^{2}\mathcal{Q}%
\left( \overline{\boldsymbol{\varphi }},\overline{\boldsymbol{\varphi }}%
\right) e^{-2\mathrm{i}k_{w}\xi }+\medskip \\ 
\hspace{1cm}+\frac{1}{2}\Omega _{T_{1}}\boldsymbol{\varphi }e^{\mathrm{i}\hat{k%
}\xi }+\frac{1}{2}\overline{\Omega }_{T_{1}}\overline{\boldsymbol{\varphi }}%
e^{-\mathrm{i}k_{w}\xi }-\frac{1}{2}\mathcal{A}_{1}\mathcal{L}_{\mathcal{A}}%
\left[ \Omega \boldsymbol{\varphi }e^{\mathrm{i}k_{w}\xi }+\overline{\Omega }%
\overline{\boldsymbol{\varphi }}e^{-\mathrm{i}k_{w}\xi }\right] .%
\end{array}
\label{eq:ordine2_wave}
\end{equation}%
The solvability condition $\langle \mathbf{F},\boldsymbol{\chi }\rangle =0$
with $\boldsymbol{\chi \in }\text{Ker}\left\{ \left[ \mathcal{L}^{\ast}\vert_{\mathcal{A}=
\mathcal{A}_{c}}+\mathrm{i}%
k_w\left( \mathcal{C}{I}-{M}\right) \right]^{\dagger }\right\} $ implies that $\mathcal{A}_{1}=0$ and that the pattern amplitude must be independent of the slow
time scale $T_{1}$, i.e. $\Omega _{T_{1}}=0$. 
The vector $\chi$ reads
\begin{equation}
\boldsymbol{\chi }=\left[ 
\begin{array}{c}
1\\ 
n_{2}+\mathrm{i}\hat{n}_{2} \\ 
n_{3}+\mathrm{i}\hat{n}_{3} \\ 
n_{4}+\mathrm{i}\hat{n}_{4}%
\end{array}%
\right] ,  \label{eq:autovettore_wave}
\end{equation}%
with 
\begin{equation}
\begin{alignedat}{2}
n_{2}&=-\frac{f_{U}^{\ast}}{g_{U}^{\ast}}, &\qquad \hat{n}_{2}&=\frac{\left( \mathcal{C}+%
\mathcal{V}\right) k_{w}}{g_{U}^{\ast}},\medskip \\ 
n_{3}&=-\frac{g_{S}^{\ast}\left( h_{S}^{\ast}n_{2}-\mathcal{C}k_{w}\hat{n}_{2}%
\right) }{{h_{S}^{\ast}}^{2}+\mathcal{C}^{2}\,k_{w}^{2} }, &\qquad \hat{n}%
_{3}&=-\frac{g_{S}^{\ast}\left( \mathcal{C}k_{w}n_{2}+h_{S}^{\ast}\hat{n}_{2}\right) 
}{ {h_{S}^{\ast}}^{2}+\mathcal{C}^{2}\,k_{w}^{2} } \medskip \\ 
n_{4}&=\frac{k_{w}\tau\left( \mathcal{C}k_{w}\tau n_{2}-\hat{n}_{2}\right) }{
\mathcal{C}^{2}\,k_{w}^{2}\tau^{2}+1 } &\qquad \hat{n}_{4}&=\frac{k_{w}\tau\left( \mathcal{C}%
k_{w}\tau \hat{n}_{2}+n_{2}\right) }{\mathcal{C}^{2}\,k_{w}^{2}\tau
^{2}+1}%
\end{alignedat}
\end{equation}
According to the previous statements, the solution of %
\eqref{eq:ordine2_wave} can be expressed as
\begin{equation}
\mathbf{W}_{2}=\Omega^{2}\mathbf{m}_{22}e^{2\mathrm{i}k_w\xi }+\overline{\Omega }%
^{2}\overline{\mathbf{m}}_{22}e^{-2\mathrm{i}k_w\xi }+|\Omega |^{2}\mathbf{m}%
_{20},
\end{equation}%
where the vectors $\mathbf{m}_{20}$ and $\mathbf{m}_{22}$ satisfy the
following systems
\begin{equation}
\begin{array}{l}
\mathcal{L}^{\ast}\vert_{\mathcal{A}=
\mathcal{A}_{c}}\mathbf{m}_{20}=-\frac{1}{4}\mathcal{Q}\left( \boldsymbol{%
\varphi },\overline{\boldsymbol{\varphi }}\right) ,\medskip  \\ 
\left[ \mathcal{L}^{\ast}\vert_{\mathcal{A}=
\mathcal{A}_{c}}+2\mathrm{i}k_w\left( \mathcal{C}{I}-{M}\right) %
\right] \mathbf{m}_{22}=-\frac{1}{8}\mathcal{Q}\left( \boldsymbol{\varphi },%
\boldsymbol{\varphi }\right) ,\medskip  \\ 
\left[ \mathcal{L}^{\ast}\vert_{\mathcal{A}=
\mathcal{A}_{c}}-2\mathrm{i}k_w\left( \mathcal{C}{I}-{M}\right) %
\right] \overline{\mathbf{m}}_{22}=-\frac{1}{8}\mathcal{Q}\left( \overline{%
\boldsymbol{\varphi }},\overline{\boldsymbol{\varphi }}\right) .%
\end{array}
\label{eq:sistemi_sol_wave_ordine2}
\end{equation}
Moreover, at the third perturbative order, System \eqref{ODE_ordini3} can be written as
\begin{equation}
\begin{array}{l}
\widehat{\mathcal{L}}\mathbf{W}_{3}=\frac{1}{2}\left\{ \Omega _{T_{2}}%
\boldsymbol{\varphi }-\mathcal{A}_{2}\Omega \mathcal{L}_{\mathcal{A}}%
\boldsymbol{\varphi }-\Omega |\Omega |^{2}\left[ \mathcal{Q}\left( 
\boldsymbol{\varphi },\mathbf{m}_{20}\right) +\mathcal{Q}\left( \overline{%
\boldsymbol{\varphi }},\mathbf{m}_{22}\right) +\frac{1}{8}\mathcal{R}\left( 
\boldsymbol{\varphi },\overline{\boldsymbol{\varphi }},\boldsymbol{\varphi }%
\right)\right]  \right\} e^{\mathrm{i}{k}_w\xi }+ \medskip \\ 
\hspace{1.1cm}+\frac{1}{2}\left\{ \overline{\Omega }_{T_{2}}\overline{%
\boldsymbol{\varphi }}-\mathcal{A}_{2}\overline{\Omega }\mathcal{L}_{%
\mathcal{A}}\overline{\boldsymbol{\varphi }}-\overline{\Omega }|\Omega |^{2}%
\left[ \mathcal{Q}\left( \overline{\boldsymbol{\varphi }},\mathbf{m}%
_{20}\right) +\mathcal{Q}\left( \boldsymbol{\varphi },\overline{\mathbf{m}}%
_{22}\right) +\frac{1}{8}\mathcal{R}\left( \overline{\boldsymbol{\varphi }},%
\overline{\boldsymbol{\varphi }},\boldsymbol{\varphi }\right) \right]
\right\} e^{-\mathrm{i}{k}_w\xi }+\medskip  \\ 
\hspace{1.1cm}-\frac{1}{2}\Omega^{3}\left\{ \mathcal{Q}\left( \boldsymbol{%
\varphi },\mathbf{m}_{22}\right) +\frac{1}{24}\mathcal{R}\left( \boldsymbol{%
\varphi },\boldsymbol{\varphi },\boldsymbol{\varphi }\right) \right\} e^{3%
\mathrm{i}{k}_w\xi }+\medskip  \\ 
\hspace{1.1cm}-\frac{1}{2}\overline{\Omega }^{3}\left\{ \mathcal{Q}%
\left( \overline{\boldsymbol{\varphi }},\overline{\mathbf{m}}_{22}\right) +%
\frac{1}{24}\mathcal{R}\left( \overline{\boldsymbol{\varphi }},\overline{%
\boldsymbol{\varphi }},\overline{\boldsymbol{\varphi }}\right) \right\} e^{-3%
\mathrm{i}{k}_w\xi },\medskip 
\end{array}%
\end{equation}%
and its solvability condition $\langle \mathbf{G},\boldsymbol{\chi }\rangle
=0$ leads to the following complex cubic Stuart-Landau equation \cite{VanSaarloosHohenberg1992, Doelman1995, Mielke1998, Mielke2002, AransonKramer2002}
\begin{equation}
\frac{\partial \Omega }{\partial T_{2}} =\left( \sigma _{1}+\mathrm{i}\sigma
_{2}\right) \Omega -\left( L_{1}-\mathrm{i}L_{2}\right) {\lvert \Omega \rvert }%
^{2}\Omega ,  \label{eq:SL_wave} 
\end{equation}%
where the coefficients here appearing are given by 
\begin{eqnarray}
\sigma _{1}+\mathrm{i}\sigma _{2} &=&\mathcal{A}_{2}\frac{\mathcal{L}_{%
\mathcal{A}}{\boldsymbol{\varphi }}\cdot {\overline{\boldsymbol{\chi }}}}{{\boldsymbol{\varphi }}\cdot {\overline{\boldsymbol{\chi }}}},\medskip\\ 
L_{1}-\text{i%
}L_{2}&=&-\frac{1}{8}\frac{\left[ \mathcal{R}\left( \boldsymbol{\varphi },%
\boldsymbol{\varphi },\overline{\boldsymbol{\varphi }}\right) +8\mathcal{Q}\left( 
\boldsymbol{\varphi },\mathbf{m}_{20}\right) +8\mathcal{Q}\left( \overline{%
\boldsymbol{\varphi }},\mathbf{m}_{22}\right) \right] \cdot {%
\overline{\boldsymbol{\chi }}}}{{\boldsymbol{\varphi }}\cdot \overline{\boldsymbol{\chi }}}.  
\label{eq:SL_coeff_wave} 
\end{eqnarray}%
As known, the dynamical regime is identified by the real part of Landau coefficient $L_{1}$. However, due to the highly nontrivial and nonlinear dependence of the above coefficients on the model parameters, numerical investigations are needed to elucidate the role of autotoxicity and inertia in vegetation pattern dynamics close to onset.


The first investigation focuses on the dependence of the dynamical regime for patterned solutions on the model parameters, a feature which is expressed through the sign of the real part of the Landau coefficient $L_1$: positive (negative) values denote supercritical (subcritical) dynamics. In Fig.~\ref{figure8} we evaluate numerically this coefficient from Eq.~\eqref{eq:SL_coeff_wave} in the $\left(\mathcal{H},\mathcal{A}\right)$-plane for different values of inertial times ($\tau=1$ and $\tau=3$). For $\tau=1$ (panel (a)), there exists a region characterised by small values of autotoxicity ($\mathcal{H}<0.02$) and large values of rainfall ($\mathcal{A}>2.2$) where the coefficient $L_1$ becomes negative giving rise to the emergence of subcritical dynamics. By increasing the inertial time further up to $\tau=3$ (panel (b)), the region of subcritical modes takes place for any value of autotoxicity strength. In particular, this region tends to shrink for increasing $\mathcal{H}$ and is observed for progressively smaller values of rainfall. From the ecological viewpoint, in the presence of strong inertial effects, the emergence of large amplitude vegetation patterns requires moderate levels of rainfall which may be gradually reduced with the increase of autotoxicity strength. 
\begin{figure}[b!]
	\centering
	\includegraphics[width=1\textwidth]{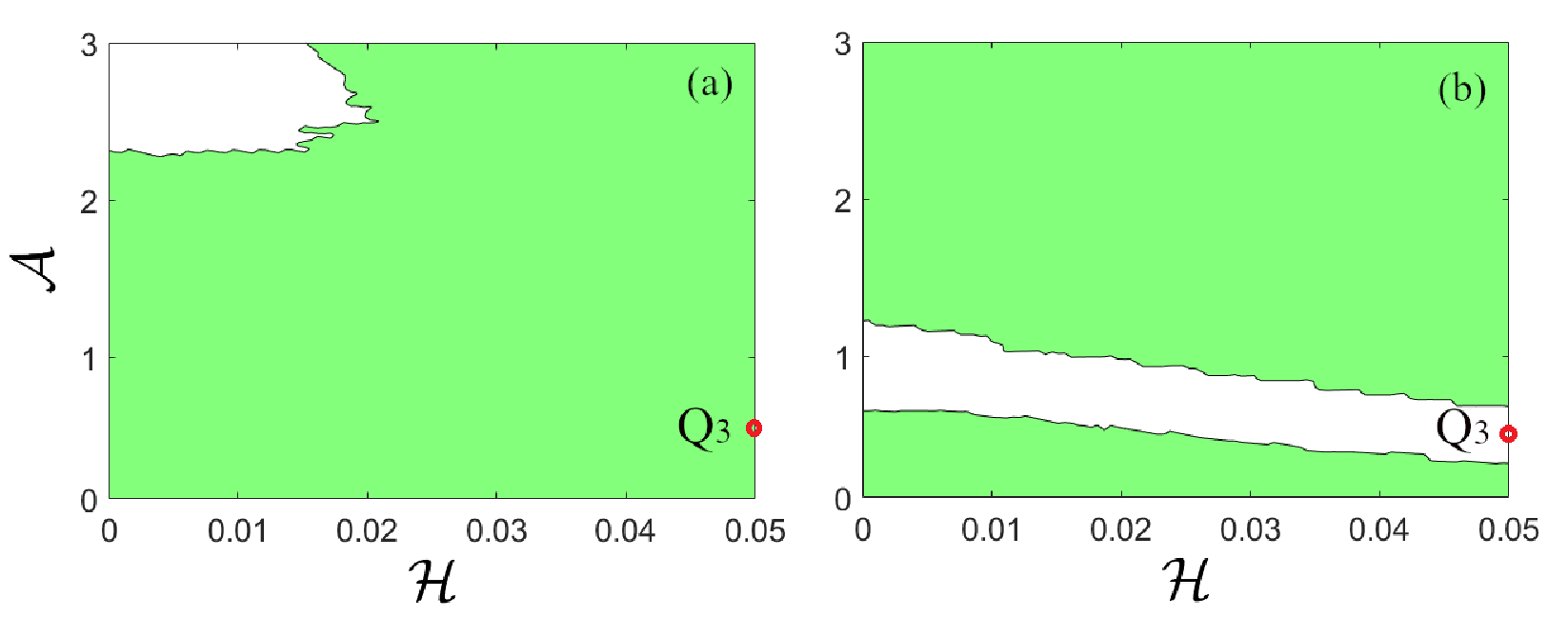}
	\caption{Landau coefficient $L_1$ obtained for (a) $\tau=1$ or (b) $\tau=3$: positive (negative) values are represented in green (white) and denote a supercritical (subcritical) regime.}	
	\label{figure8}	
\end{figure}

To validate this theoretical prediction, we inspect further aspects of the underlying phenomenon, as summarised in Fig.~\ref{figure9}. Here, we consider the parameter set corresponding to point $\mathrm{Q}_3=\left(\mathcal{H},\mathcal{A}\right)=(0.05,0.5)$ depicted in Fig.~\ref{figure8} and perform additional numerical investigations with $\tau=1$ (left panels) and $\tau=3$ (right panels). A time window $t\in[0,1000]$ is here implemented. At first, simulations make use of a small random perturbation of the steady-state $\mathbf{W}_{R}^{\ast}$ and the control parameter is set slightly above threshold by fixing $\epsilon^2=0.1$ and $\mathcal{A}_2=\mathcal{A}_c$ in \eqref{eq:sviluppo}. 
This choice provides the two patterned solutions shown in the first row of Fig.~\ref{figure9}. As it can be noticed, these solutions appear significantly different for a twofold reason. On the one hand, the width of the vegetated portion of the tiger bush is much smaller far from the parabolic limit (compare the extension of the green regions in the two figures with respect to the yellow ones). On the other hand, considerably larger variations of the pattern amplitude are observed far from the parabolic limit (see the different scales in the colorbar). These constitute a preliminary indication on the different dynamical regime in which such patterns emerge. 
\begin{figure}[t!]
	\centering
	\includegraphics[width=1\textwidth]{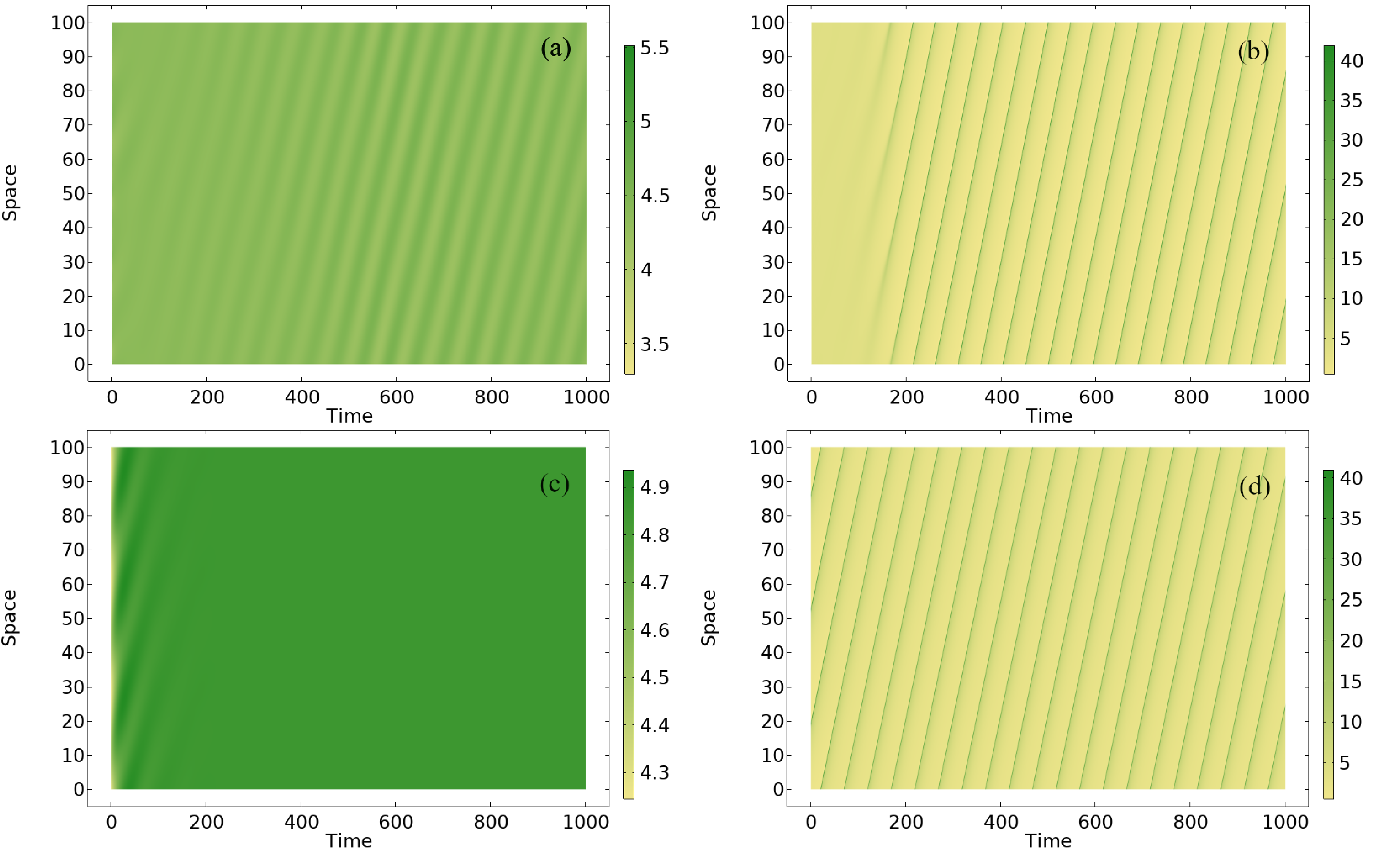}
	\caption{(a,b) Vegetation patterns obtained by using the parameter set corresponding to point $\mathrm{Q}_3=\left(\mathcal{H},\mathcal{A}\right)=(0.05,0.5)$ depicted in Fig.~\ref{figure8} for different values of inertial time (left column: $\tau=1$, right column: $\tau=3$). Here, the initial condition is a small random perturbation of the steady-state $\mathbf{W}_{R}^{\ast}$ and the control parameter is set slightly above threshold by an amount of $\epsilon^2=0.1$.
(c,d)  Vegetation patterns obtained using the final configuration ($t=1000$) depicted in (a,b) as initial condition for the new simulation where the control parameter is set slightly below threshold by an amount of $\epsilon^2=0.1$}.	
	\label{figure9}	
\end{figure}
To confirm it, we carry out another investigation where the final configuration obtained in panels (a) and (b) is taken as initial condition for a new simulation where the control parameter is, now, set slightly below the threshold by an amount of $\epsilon^2=0.1$. Results corroborate the validity of our predictions: for $\tau=1$, the patterned solutions progressively die out and the system converges towards the uniform vegetated state (panel (c)) whereas, for $\tau=3$, the patterned state still survives. This is a clear evidence that, in the former case, dynamics take place in the supercritical regime, as the branch of patterned solution originates gradually from the homogeneous (zero-amplitude) state and no stable solutions exist for below-threshold values of the control parameter. On the contrary, in the latter case, dynamics fall into the subcritical regime since, even close to onset, patterns exhibit large amplitude and, below onset, they still exist denoting a hysteretic behaviour. Therefore, it can be safely concluded that the point $\mathrm{Q}_3$ selected in Fig.~\ref{figure8}(a) and (b) lies in the supercritical and subcritical regime, respectively. This result is also in line with our previous theoretical investigations \cite{Consolo2022PRE, Consolo2022III, Grifo2025II}, where it was pointed out that inertia may not only act as a simple time lag but can also affect the regime in which patterned dynamics take place.

\section{Travelling vegetation pulses far from onset} \label{sec:travpulse}

In the following, consistently with ecologically relevant scenarios, we consider {$\mathcal{V}=1/\varepsilon$} with $0 < \varepsilon \ll 1$. System \eqref{eq:modadim} can thus be rewritten as
\begin{equation} \label{eq:modadimeps}
\begin{aligned}
    U_t-\varepsilon^{-1}U_x &= \mathcal{A}-U-U V^2, \\
    V_t+J_x &= U V^2-\mathcal{B}V-\mathcal{H}SV, \\
    \mathcal{D} S_t &= \mathcal{B} V + \mathcal{H}S V-S, \\
    \tau J_t +V_x &= -J.
\end{aligned}
\end{equation}
\noindent
The steady-states associated to Eq.~\eqref{eq:modadimeps} are given by $\left( U^\ast, V^\ast, S^\ast, J^\ast \right) = \left( \mathcal{A}, 0,0,0 \right)$ representing the desert state and two nontrivial homogeneous steady-states as given in \eqref{equilibria}.

The goal of this section is to investigate the influence of the inertial parameter $\tau$ on the speed $\mathcal{C}>0$ and the structure of travelling pulse solutions originating from the desert state observed numerically in the far-from-equilibrium regime (see e.g.~Fig.~\ref{fig:num_GSPT}). In particular, we aim at constructing these pulses using the methods of GSPT and compare them with the ones studied in the parabolic case ($\tau=0$) in \cite{grifo2025far}. We remark that, as highlighted in \cite{Carter2018}, periodic wave train solutions can be constructed by concatenating individual vegetation stripes on a one-dimensional domain; because of this reason, and consistently with \cite{grifo2025far}, we will focus our subsequent analysis on single travelling pulses.\\
To set the stage for our analytical construction, we first provide numerical evidence illustrating the existence and qualitative features of travelling pulses in the hyperbolic regime. To this goal, direct numerical simulations of System \eqref{eq:modadimeps} are performed in COMSOL Multiphysics\textsuperscript{\textregistered} \cite{Comsol} over a spatial domain of length $L=2000$ with periodic boundary conditions. Initial conditions are set as constant for water $U$ and autotoxicity $S$, namely $U(x,0)=0.5$ and $S(x,0)=0$ for all $x\in[0,L]$, whereas for the vegetation variable $V$ we prescribe a Gaussian pulse centered at $x=300$ with standard deviation $\sigma = 25$ and amplitude $V_{\max} = 10$. Model parameters are fixed as $\mathcal{A} = 1.2$, $\mathcal{B} = 0.45$, $\varepsilon = 0.005$, $\mathcal{D} = 4.5$, and $\mathcal{H} = 1$. Results are illustrated in Fig.~\ref{fig:num_GSPT}. In particular, panels~(a)–(c) report the spatial profiles of surface water $U$, biomass $V$, and toxicity $S$ for the travelling pulse obtained for the two values of the inertial parameter $\tau=0$ and $\tau=0.019$. These profiles highlight the nonlinear character of the solution and the coexistence of multiple spatial scales. Moreover, it can be observed that increasing $\tau$ slightly affects the amplitude of the vegetation and toxicity components while leaving the overall shape of all profiles essentially unchanged. On the other hand, panels~(d)–(e) display the corresponding space–time evolution of the biomass density, confirming the existence of a localised peak propagating at constant speed across the domain. A quantitative comparison of the migration speeds reveals a clear influence of inertia: for $\tau=0$ the numerically pulse speed is $\mathcal{C}_{\mathrm{num}}^{(0)} = 5.84$, whereas for $\tau=0.019$ it increases to $\mathcal{C}_{\mathrm{num}}^{(0.019)} = 6.08$.
Analogously to the parabolic case, here we also have that the state variables evolve on three scales -- defined as \emph{superslow}, \emph{slow} and \emph{fast} -- which are to be intended in terms of the travelling wave variable rather than in the temporal sense. 
However, moving from the parabolic to the hyperbolic framework leads to a change of monotonicity in the wave speed (consistently with numerical observations) and also requires a more complex analysis of the pulse' dynamics in phase space. Our main analytical result consists in the following

\begin{theorem} \label{thm:main}
    For $0 < \varepsilon \ll 1$ sufficiently small and $0 < \tau < \tau_{\rm max} \leq \tau_{\rm upper}$,
    there exists a unique $\theta(\tau)>0$ such that there exists a travelling wave solution $\left( U, V, S, J \right) (x,t) = \left( U, V, S, J \right) (x-\mathcal{C} t)$ of System \eqref{eq:modadim} with wave speed $\mathcal{C}=\left( \frac{\mathcal{A}^2 \theta(\tau)^2}{\varepsilon} \right)^{1/3} + \mathcal{O}(1)=:\mathcal{C}^\ast(\tau)$.
\end{theorem}

\noindent
Further details regarding both the function $\theta(\tau)$ as well as $\tau_{\max}$, $\tau_{\rm upper}$ will be provided in Sec.~\ref{sec:persist}. We remark here that for the parabolic case we retrieve $\theta(0)\approx0.8615=\theta_0$, consistently with \cite{Carter2018, grifo2025far}.


\begin{figure}[t!]
	\centering
	\includegraphics[width=1\textwidth]{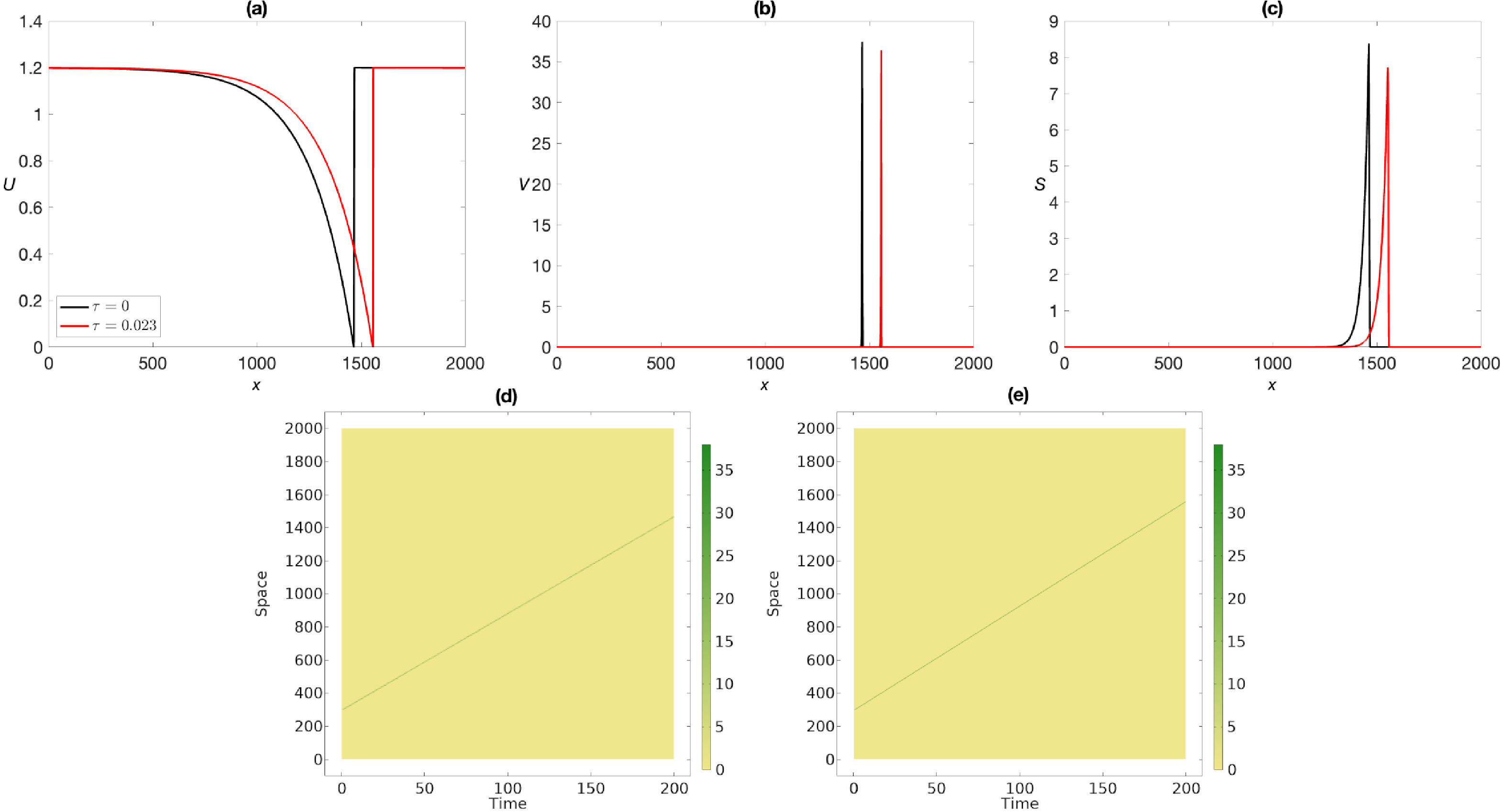}
	\caption{Spatial profiles for surface water $U$ (a), biomass $V$ (b) and toxicity $S$ (c) of a travelling pulse, obtained by numerical integration of System \eqref{eq:modadim} over a domain with length $L=2000$ with periodic boundary conditions, for different values of $\tau$. Spatio-temporal dependence of the biomass density $V(x,t)$ for $\tau = 0$ (panel (d)) and $\tau = 0.019$ (panel (e)). The parameter values are $\mathcal{A}=1.2$, $\mathcal{B}=0.45$, $\varepsilon=0.005$, $\mathcal{D}=4.5$, and $\mathcal{H}=1$.}
	\label{fig:num_GSPT}
\end{figure}

\subsection{Existence of travelling pulses}

In order to construct our sought-after travelling pulse with fixed speed $\mathcal{C}>0$, we introduce the comoving frame variable $\xi=x-\mathcal{C}t$, which transforms Eq.~\eqref{eq:modadimeps} in the following system of first-order ODEs
\begin{equation} \label{eq:prescaling}
\begin{aligned}
    U_\xi &= \frac{\varepsilon}{1 + \varepsilon \mathcal{C}} \left(U - \mathcal{A} + U V^2 \right), \\
    V_\xi  &= \frac{1}{1-\mathcal{C}^2 \tau} \left( \mathcal{C} \tau \left( U V^2-\mathcal{B}V-\mathcal{H}SV \right)-J \right), \\
    S_\xi &= \frac{1}{\mathcal{C} \mathcal{D}} \left(S - \mathcal{B} V - \mathcal{H} S V\right), \\
    J_\xi &= \frac{1}{1-\mathcal{C}^2 \tau} \left( U V^2-\mathcal{B}V-\mathcal{H}SV-\mathcal{C}J \right).
\end{aligned}
\end{equation}
We apply the same rescaling carried out in \cite{grifo2025far} to the variables $ V, S, J$ and the speed $\mathcal{C}$, while considering $U, \mathcal{A}, \mathcal{B}, \mathcal{H}, \mathcal{D}$ to be $\mathcal{O}(1)$. However, we additionally introduce a rescaling on $\tau$ in order to capture the pulse different spatial scales observed numerically (see Fig.~\ref{fig:num_GSPT}). This finally leads to
\begin{equation}  \label{eq:rescaling}
\begin{array}{c c c c}
U = u,& V = c^{-1} \varepsilon^{-2/3} v,& 
J = \varepsilon^{-1} j,&
S = c^{-1} \varepsilon^{-2/3} s,  \\
\mathcal{C} = c\, \varepsilon^{-1/3},&
\xi = c^2 \varepsilon^{1/3} \zeta,&
\delta = c\, \varepsilon^{2/3}, & \tau = c\, \varepsilon^{2/3} \hat{\tau}.
\end{array}
\end{equation}
In these rescaled variables, and indicating $\dot{}=\frac{d}{d\zeta}$, System \eqref{eq:prescaling} becomes
\begin{equation} \label{eq:postscaling}
\begin{aligned}
\dot{u} &= \frac{1}{1+\delta} \left( u v^2 + \delta^2 \left( u - \mathcal{A} \right) \right), \\
\dot{v} &= \frac{c^3}{1-c^3\hat{\tau}} \left( \hat{\tau} \left( u v^2-\delta\mathcal{B}v - \mathcal{H}sv \right) -j \right),  \\
\dot{s} &= \frac{1}{\mathcal{D}} \left(\delta s - \delta \mathcal{B} v - \mathcal{H} s v\right), \\
\dot{j} &= \frac{1}{1-c^3\hat{\tau}} \left(u v^2- \delta \mathcal{B} v - \mathcal{H} s v - c^3 j \right),
\end{aligned}
\end{equation}
where the small parameter which will be considered in the subsequent GSPT analysis is now $\delta$. The travelling pulse still corresponds to a homoclinic orbit to the equilibrium $(u, v, s, j) = \left( \mathcal{A},0,0,0 \right)$, corresponding to the desert state $\left( U^\ast, V^\ast, S^\ast, J^\ast \right)$. The scale separation in System \eqref{eq:postscaling} is made more evident by replacing the $u$-dynamics with the ones for the new variable
\begin{equation}\label{eq:definition_w_variable}
w = \left( 1 + \delta \right) u - j + v + \mathcal{D} s.
\end{equation}
This finally leads to the \emph{fast} system
\begin{equation} \label{eq:postscalingw}
\begin{aligned}
\dot{w} &= \delta s + \frac{\delta^2}{1+\delta} \left( w +j - v - \mathcal{D} s - a \right), \\
\dot{v} &= \frac{c^3}{1-c^3\hat{\tau}} \left( \frac{\hat{\tau}}{1+\delta} \left( w+j-v-\mathcal{D}s \right) v^2 -\delta\hat{\tau}\mathcal{B}v - \hat{\tau}\mathcal{H}sv -j \right), \\
\dot{s} &= \frac{1}{\mathcal{D}} \left(\delta s - \delta \mathcal{B} v - \mathcal{H} s v\right), \\
\dot{j} &= \frac{1}{1-c^3\hat{\tau}} \left( \frac{1}{1+\delta} \left( w+j-v-\mathcal{D}s \right) v^2 - \delta \mathcal{B} v - \mathcal{H} s v - c^3 j \right),
\end{aligned}
\end{equation}
where $a = \left( 1 + \delta \right) \mathcal{A}$. The (equivalent) \emph{slow} and \emph{superslow} formulation of System \eqref{eq:postscalingw} will be provided in Sec.~\ref{sec:redprob}. The sought-after travelling pulse thus corresponds to a homoclinic orbit to the steady-state $(w^*,v^*,s^*,j^*)=(a,0,0,0)$.
The presence of the small parameter $0 < \delta \ll 1$ allows us to obtain such homoclinic solutions as perturbations of singular orbits obtained in the limit $\delta=0$ by matching slow dynamics on the critical manifold (i.e., a manifold of equilibria of System \eqref{eq:postscalingw} for $\delta=0$) with fast excursions along heteroclinic orbits emerging in the corresponding layer problem, analysed in Sec.~\ref{sec:laypb}. The singular solutions will be constructed in Sec.~\ref{sec:singorb}, and their persistence will be proved in Sec.~\ref{sec:persist}. In Sec.~\ref{sec:numcont}, the software AUTO \cite{Doedel1981} will be employed in order to perform numerical continuation of the travelling pulse solution w.r.t.~the main system parameters.

\begin{remark}
    In the following, we will focus on the assumption $1-c^3\hat{\tau}>0$, as we are continuing travelling pulses from the parabolic case corresponding to $\hat{\tau}=0$. This leads to the (implicit) constraint $\hat{\tau} < \frac{1}{c^3}$, which in terms of the original variables corresponds to $\tau < \frac{1}{\mathcal{C}^2}$. We point out that in the analysis below several rescaled quantities (including e.g.~$c$ and $\hat{\tau}$), which are treated as parameters, are implicitly dependent on other parameters -- a feature that is often hidden at leading order. In the proof of Thm.~\ref{thm:main} (in Sec.~\ref{sec:persist}), reversing back to the original scaling will allow us to reveal a clearer link between the original variables and parameters, including an expression for both $\mathcal{C}^\ast$ and $\tau_{\rm upper}$. 
\end{remark}


\subsubsection{Critical manifolds and reduced dynamics} \label{sec:redprob}
Considering the singular limit $\delta=0$ in Eq.~\eqref{eq:postscalingw} yields
\begin{equation} \label{eq:laypb}
\begin{aligned}
\dot{w} &= 0, \\
\dot{v} &= \frac{c^3}{1-c^3\hat{\tau}} \left( \hat{\tau} \left( w+j-v-\mathcal{D}s \right) v^2 - \hat{\tau}\mathcal{H}sv -j \right), \\
\dot{s} &= -\frac{\mathcal{H}}{\mathcal{D}} s v, \\
\dot{j} &= \frac{1}{1-c^3\hat{\tau}} \left( \left( w+j-v-\mathcal{D}s \right) v^2 - \mathcal{H} s v - c^3 j \right),
\end{aligned}
\end{equation}
i.e., the layer problem. The critical manifold, obtained by determining the steady-states of the layer problem, is composed of the following three parts
\begin{subequations}\label{eq:critmanifolds}
\begin{align}
\mathcal{M}^{(1)} &= \left\{\,v=0,\, j=0,\, s=0\,\right\},\\
\mathcal{M}^{(2)} &= \left\{\,v=w,\, j=0,\, s=0\,\right\},\\
\mathcal{M}^{(3)} &= \left\{\,v=0,\, j=0\,\right\}.
\end{align}
\end{subequations}
As in the parabolic case analysed in \cite{grifo2025far}, these manifolds can be written as $\mathcal{M}^{(i)} = \bigcup_{w \in \mathbb{R}} p_i(w)$, $i=1,2,3$, where
\begin{equation}
p_1(w) = (w, 0, 0, 0),
\hspace{0.3cm}
p_2(w) = (w, w, 0, 0),
\hspace{0.3cm}
p_3(w,s) = (w, 0, s, 0).
\label{eq:equilibria}
\end{equation}
The manifolds $\mathcal{M}^{(1)}$ and $\mathcal{M}^{(3)}$ are invariant under the flow of \eqref{eq:laypb}, differently from $\mathcal{M}^{(2)}$ which is involved in the construction of the pulse only through the point $p_2(a)$ where two heteroclinic connections meet (see Sec.~\ref{sec:laypb}). Each point on $\mathcal{M}^{(2)}$ (except for the origin) is in fact a saddle with eigenvalues $-c^3$, $-(1-c^3 \hat{\tau}) \frac{\mathcal{H} w}{\mathcal{D}}$, and $(1-c^3 \hat{\tau})w^2$ in the normal directions. As for $\mathcal{M}^{(3)}$, this is also normally hyperbolic since the equilibria $p_3(w,s)$ are saddles in the $(v,j)$-directions.\\
In order to study the dynamics on $\mathcal{M}^{(3)}$, we introduce the slow scale $\sigma = \delta \zeta$, leading to
\begin{subequations}
\begin{align}
w_\sigma &= s + \frac{\delta}{1+\delta} \left( w +j - v - \mathcal{D} s - a \right), \\
\delta v_\sigma &= \frac{c^3}{1-c^3\hat{\tau}} \left( \frac{\hat{\tau}}{1+\delta} \left( w+j-v-\mathcal{D}s \right) v^2 -\delta\hat{\tau}\mathcal{B}v- \hat{\tau}\mathcal{H}sv -j \right), \\
\delta s_\sigma &= \frac{1}{\mathcal{D}} \left( \delta s - \delta \mathcal{B} v - \mathcal{H} s v \right), \\
\delta j_\sigma &= \frac{1}{1-c^3\hat{\tau}} \left( \frac{1}{1+\delta} \left( w+j-v-\mathcal{D}s \right) v^2 - \delta \mathcal{B} v - \mathcal{H} s v - c^3 j \right).
\label{eq:slow}
\end{align}
\end{subequations}
In particular, on $\mathcal{M}^{(3)}$ these reduce to linear slow-fast dynamics as follows
\begin{subequations}
\begin{align}
w_\sigma &= s + \frac{\delta}{1+\delta} \left( w - \mathcal{D} s - a \right), \\
s_\sigma &= \frac{s}{\mathcal{D}}.
\end{align}
\label{eq:slowflow_M3}
\end{subequations}
As in the parabolic case, the manifold $\mathcal{M}^{(3)}$ can thus be foliated by one-dimensional manifolds $\mathcal{M}^{(3),w_0}:= \left\{ p_3\left(w, s(w)\right) \, : \, w_0 \in \mathbb{R}_0^+ \right\}$ where $s(w):=\frac{w-w_0}{\mathcal{D}}$. On the manifold $\mathcal{M}^{(1)} \subset \mathcal{M}^{(3)}$, corresponding to the critical manifold of Eq.~\eqref{eq:slowflow_M3}, the flow is given by
\begin{equation}\label{eq:superslowM1}
 w_\eta = \frac{1}{1+\delta} \left( w - a \right),
\end{equation}
where $\eta := \delta \sigma = \delta^2 \zeta$ corresponds to the superslow variable. These dynamics simply describe a flow away from the unique unstable equilibrium at $w=a$.

\subsubsection{Layer problem} \label{sec:laypb}
Analogously to the parabolic case analysed in \cite{grifo2025far}, in the singular limit the sought-after travelling pulse corresponds to the union of a superslow orbit from $p_1(a)=(a,0,0,0)$ to the origin, a slow transition on $\mathcal{M}^{(3),0}$ until $p_3\left( \frac{a}{\mathcal{D}} \right) = \left( a,0,\frac{a}{\mathcal{D}},0 \right)$, and a fast transition from $p_3\left( \frac{a}{\mathcal{D}} \right)$ to $p_1(a)$. This last piece, in particular, consists of two heteroclinic connections from  $p_3\left( \frac{a}{\mathcal{D}} \right)$ to $p_2(a)$ and from $p_2(a)$ to $p_1(a)$, which will be analysed in Lemmas \ref{lemma:fc2}-\ref{lemma:fc1}. The inertial parameter $\hat{\tau}$, in fact, does not affect the geometric construction of the pulse in phase-space, but influences its speed, as it will be shown in Lemma \ref{lemma:fc1}.

\begin{lemma} \label{lemma:fc2}
For each $\hat{\tau}$ such that $0 < \hat{\tau} < \mathrm{min} \left\{ \frac{1}{c^3}, \frac{1}{a^2} \right\}$, there exists a unique $c(\hat{\tau})=c^\ast(\hat{\tau})>0$ for which System \eqref{eq:laypb} admits a heteroclinic orbit from $p_2(a) =(a,a,0,0)$ to $p_1(a)=(a,0,0,0)$. This heteroclinic orbit lies in the hyperplane $\left\{ (w,v,s,j) \, : \, w=a, \ s=0 \right\}$. In particular,
\begin{equation}\label{eq:c_ast}
c^\ast(\hat{\tau})=(a \, \theta (a,\hat{\tau}))^{2/3}
\end{equation}
with $\theta(a,\hat{\tau}) \approx \theta_0 +\theta_1 a^2\hat{\tau} + \theta_2 a^4\hat{\tau}^2$ where $\theta_0 \approx 0.8615$, $\theta_1 \approx 0.0183$, $\theta_2 \approx 0.0120$. 
\end{lemma}

\begin{proof}
On the hyperplane $s=0$, Eq.~\eqref{eq:laypb} becomes ($w$, $s$ are here constant)
\begin{equation} \label{eq:laypbs0}
\begin{aligned}
\dot{v} &= \frac{c^3}{1-c^3\hat{\tau}} \left( \hat{\tau} \left( w+j-v \right) v^2 -j \right), \\
\dot{j} &= \frac{1}{1-c^3\hat{\tau}} \left( \left( w+j-v \right) v^2 - c^3 j \right).
\end{aligned}
\end{equation}
Our first goal is to show that for each $w$ there exists a value $c^\ast(w, \hat{\tau})$ such that a heteroclinic orbit $\phi_{f,1}(w)$ between $p_2(w)$ and $p_1(w)$ exists. To this aim, we start by observing that for each $w$ the point $p_1(w)$ admits one strong stable direction (identified by the eigenvector $(1,1)$ associated to the eigenvalue $-\frac{c^3}{1-c^3\hat{\tau}}$) and a neutral direction (identified by the eigenvector $(1,0)$ associated to the null eigenvalue). In particular, the strong stable direction $\mathcal{W}^{ss}(p_1(w))$ admits the expansion
\begin{equation} \label{eq:Wssp1w}
    \mathcal{W}^{ss}(p_1(w)) = \left\{ (w,v,s,j) \, : \, s=0, \, j = v-\frac{w}{2 c^3}(1-c^3 \hat{\tau}) v^2 + \mathcal{O}(v^3) \right\},
\end{equation}
i.e.~lies in the first quadrant below the line $j=v$ and asymptotically approaches $p_1(w)$ along this line. Our region of interest $\mathcal{T}$ in the construction of the heteroclinic connection is in the nonnegative $(v,j)$-plane delimited by the following three curves (see Fig.~\ref{fig:qualitative_plot}):
\begin{equation}
    B_1 = \left\{ j=\frac{\hat{\tau}(w-v)v^2}{1-\hat{\tau}v^2} \right\}, \quad B_2 = \left\{ v=w \right\}, \quad B_3 = \left\{ j=v \right\}.
\end{equation}

\begin{figure}[t!]
    \centering
    \begin{overpic}[width=0.6\linewidth]{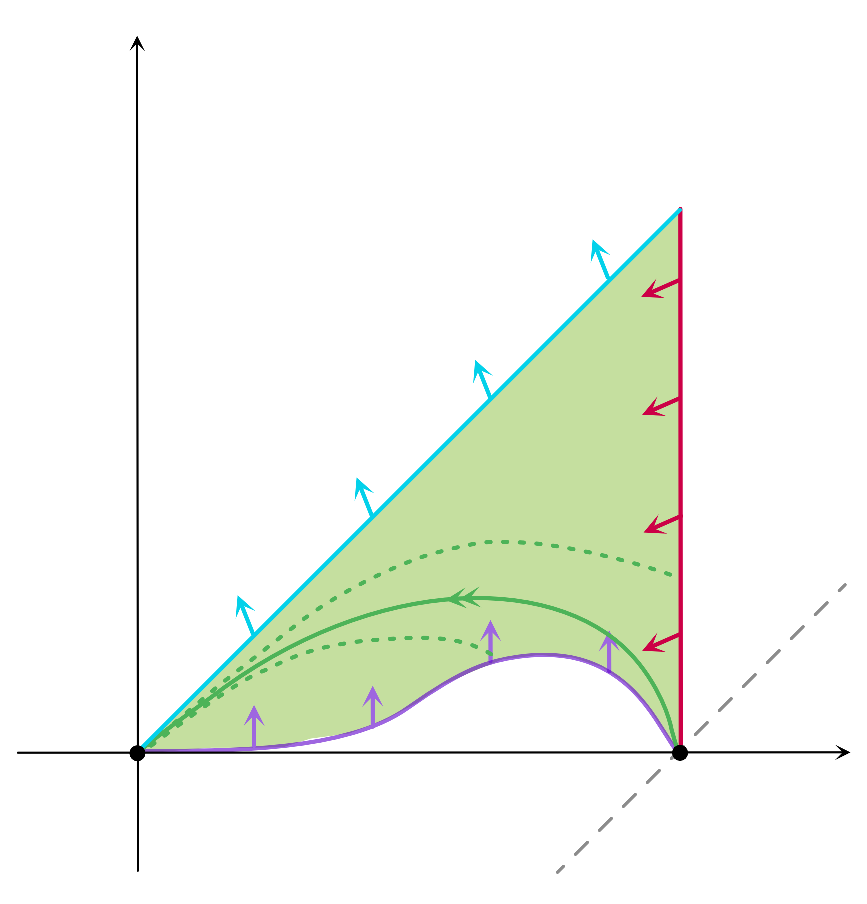}
    \put(58,24){\textcolor{myviolet}{$B_{1}$}}
    \put(76,45){\textcolor{myred}{$B_{2}$}}
    \put(44,52){\textcolor{myblue}{$B_{3}$}}
    \put(95,35){\textcolor{gray}{$v-j=w$}}
    \put(58,35){\textcolor{mygreen}{$\phi_{f,1}$}}
    \put(65,60){$\mathcal{T}$}
    \put(94,13){$v$}
    \put(12,97){$j$}
    \put(5,13){$p_1(w)$}
    \put(75,13){$p_2(w)$}
    \end{overpic}
    \caption{Sketch of the region $\mathcal{T}$ in the $(v,j)$-plane defined in the proof of Lemma \ref{lemma:fc2} to construct a heteroclinic orbit between $p_2(w)$ and $p_1(w)$.}
    \label{fig:qualitative_plot}
\end{figure}
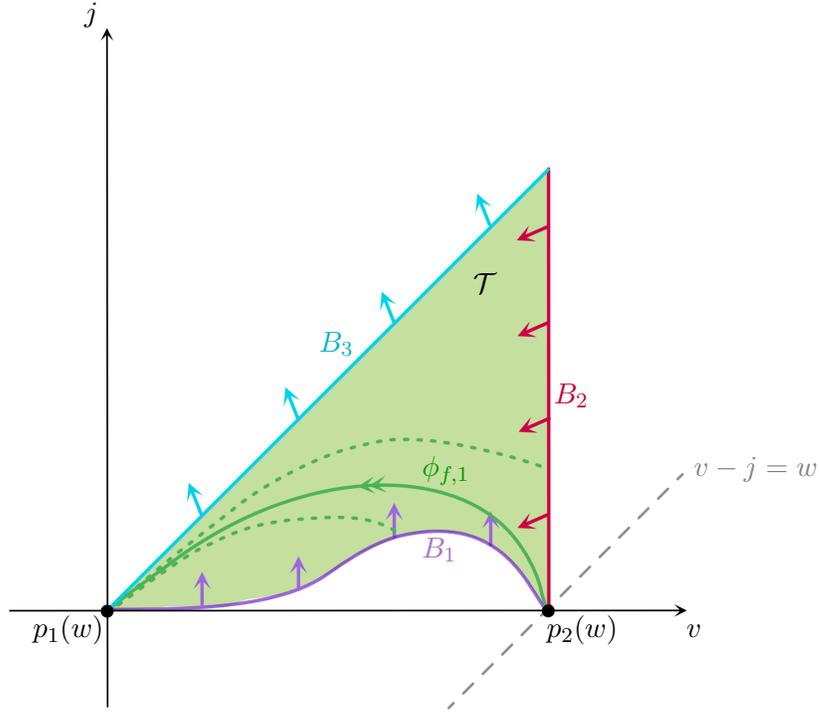
Let us examine the dynamics on each of these curves. On $B_3$ we have
\[
\dot{v}-\dot{j} = \frac{1}{1-c^3\hat{\tau}} \left( wv^2(c^3\hat{\tau}-1) \right) < 0, 
\]
hence on this curve the flow brings the orbits away from our region $\mathcal{T}$, confirming that $\mathcal{W}^{ss}(p_1(w))$ lies to the right of $B_3$. On $B_2$, we have $\dot{v}<0$ if and only if $\hat{\tau} < \frac{1}{w^2}$, whereas $\dot{j} \gtrless 0 $ if and only if $w^2 \gtrless c^3$. Finally, on $B_1$ we have $\dot{v}=0$ and $\dot{j}=\frac{j}{\hat{\tau}} > 0$. Therefore, in backward ``time'' the strong stable manifold $\mathcal{W}^{ss}(p_1(w))$ only has three possibilities: either it exits along $B_1$ or $B_2$, or precisely at $p_2(w)$, thus leading to the desired heteroclinic connection. We remark that the line $v-j=w$ is invariant for the flow \eqref{eq:laypbs0}, since on this line $\dot{v}-\dot{j}=0$; consequently, when trajectories leave either via $B_1$ or $B_2$, they remain confined in the regions $\left\{ v+w < j < \frac{\hat{\tau}(w-v)v^2}{1-\hat{\tau}v^2}  \right\}$ and $\left\{ w < v < j + w \right\}$, respectively. \\
The idea is then to connect a backward exit from $B_1$ (respectively $B_2$) to small (respectively large) speeds $c$; by continuity w.r.t.~$w$, this will guarantee the existence of a speed $c=c^\ast(w, \hat{\tau})$ for each $\hat{\tau}$ within the admissible range such that $\mathcal{W}^{ss}(p_1(w))$ connects with $p_2(w)$ in backward ``time''. To this aim, we observe that from Eq.~\eqref{eq:laypbs0} we have
\begin{equation} \label{eq:djdvineq}
\frac{dj}{dv} = 1-\frac{v^2}{c^3}(1-c^3 \hat{\tau})\frac{w+j-v}{j-\hat{\tau}(w+j-v) v^2} < 1-\frac{v^2}{c^3}(1-c^3 \hat{\tau})
\end{equation}
Consequently, the solution $j(v)$ to Eq.~\eqref{eq:laypbs0} always lies below the curve $\bar{j}(v)$, where $\bar{j}(v)$ is the solution to the equation
\begin{equation}
\frac{d\bar{j}}{dv} =  1-\frac{v^2}{c^3}(1-c^3 \hat{\tau}) \quad \Leftrightarrow \quad \bar{j}(v) = v \left( 1-\frac{1-c^3 \hat{\tau}}{3 c^3} v^2 \right).
\end{equation}
If $\bar{j}(w) < 0$, then the solution $j(v)$ must leave in backward ``time'' $\mathcal{T}$ through $B_1$. This condition is satisfied when
\begin{equation} \label{eq:condspB1}
    c^3 < \frac{w^2}{3+\hat{\tau} w^2}.
\end{equation}
As for $B_2$, we proceed in a similar manner by determining a curve $\hat{j}(v)$ such that $j(v) > \hat{j}(v)$ in $\mathcal{T}$ and that $\hat{j}(w)>0$. To this aim, we preliminarily observe from Eq.~\eqref{eq:djdvineq} that
\begin{equation} \label{eq:djdvineq2}
\begin{aligned}
&\frac{dj}{dv} = 1+\frac{v^3(1-c^3 \hat{\tau})}{c^3 \left( j-\hat{\tau}(w+j-v) v^2 \right)}-\frac{v^2}{c^3}(1-c^3 \hat{\tau})\frac{w+j}{j-\hat{\tau}(w+j-v) v^2} > \frac{v^3(1-c^3 \hat{\tau})}{c^3 \left( j-\hat{\tau}(w+j-v) v^2 \right)} \\
& \,  \Leftrightarrow \, c^3 > v^2 (1-c^3 \hat{\tau}) \frac{w+j}{j-\hat{\tau}(w+j-v) v^2} \, \Leftrightarrow \, j(c^3-v^2) > v^2(w-c^3 \hat{\tau} v).
\end{aligned}
\end{equation}
Taking $c^3 > 2w^2$ and $ j > \frac{2 w v^2}{c^3}$ ensures that $c^3-v^2>0$ and $j > \frac{v^2(w-c^3 \hat{\tau} v)}{c^3-v^2}$, thus guaranteeing the validity of the last inequality in \eqref{eq:djdvineq2}. Therefore, under these conditions we have
\begin{equation}
\frac{dj}{dv} > \frac{v^3(1-c^3 \hat{\tau})}{c^3 \left( j-\hat{\tau}(w+j-v) v^2 \right)} > \frac{v^3}{c^3 j}(1-c^3 \hat{\tau}).
\end{equation}
We thus have that the solution $j(v)$ to Eq.~\eqref{eq:laypbs0} always lies above the curve $\hat{j}(v)$ which, together with the conditions $\hat{j}(0)=0$ and $\hat{j} > 0$, solves
\begin{equation}
\frac{d\hat{j}}{dv} = \frac{v^3}{c^3 j}(1-c^3 \hat{\tau}) \quad \Leftrightarrow \quad \hat{j}(v) = v^2 \sqrt{\frac{1-c^3 \hat{\tau}}{2 c^3}}.
\end{equation}
This function indeed satisfies $\hat{j}(w) > 0$. However, in order to ensure that $j(w) > \frac{2 w^3}{c^3}$, we further require that $c^3 (1-c^3 \hat{\tau}) > 8 w^2$, which in turn leads to
\begin{equation}
    \frac{1}{2 \hat{\tau}} \left( 1-\sqrt{1-32 \hat{\tau} w^2} \right) < c^3 < \frac{1}{2 \hat{\tau}} \left( 1+\sqrt{1-32 \hat{\tau} w^2} \right).
\end{equation}
When the speed lies within this range, the orbit $j(v)$ leaves $\mathcal{T}$ in backward ``time'' through $B_2$.  In particular, this leads to the constraint $\hat{\tau} < \frac{1}{32 w^2}$ -- which is, however, not restrictive (see Rem.~\ref{rem:constrtau}). Then, by continuity in $w$, there exists a value $c^\ast(w, \hat{\tau})$ in the range
\begin{equation} \label{eq:rangecw}
\frac{w^2}{3+\hat{\tau} w^2} < (c^\ast(w,\hat{\tau}))^3 < \frac{1}{2 \hat{\tau}} \left( 1-\sqrt{1-32 \hat{\tau} w^2} \right)
\end{equation}
such that $\mathcal{W}^{ss}(p_1(w))$ reaches $p_2(w)$ in backward ``time'', leading to a heteroclinic connection between $p_1(w)$ and $p_2(w)$ indicated by $\phi_{f,1}(w)$. The uniqueness of such a speed is now proved via Melnikov integrals by showing that the splitting between $\mathcal{W}^{ss}(p_1(w))$ and the unstable manifold $\mathcal{W}^{u}(p_2(w))$ along the heteroclinic connection $\phi_{f,1}(w)$, represented by the distance function $D(\tilde{c}; w, \hat{\tau})$ with $\tilde{c}=c-c^\ast(w,\hat{\tau})$, has a unique zero to leading order at $\tilde{c}=0$ for each value of $w$ and $\hat{\tau}$ in its admissibility range. Such function is given by
\begin{equation} \label{eq:distc}
    D(\tilde{c}; w, \hat{\tau}) = M_c(w,\hat{\tau}) \tilde{c}(w) + \mathcal{O}(\tilde{c}^2(w)),
\end{equation}
where $M_c(w,\hat{\tau})$ is the Melnikov integral associated to Eq.~\eqref{eq:laypbs0} as follows. The adjoint equation associated to the linearisation of the layer problem \eqref{eq:laypbs0} about the front $\phi_{f,1}(\zeta; w) = \left( v_{f,1}(\zeta; w), j_{f,1}(\zeta; w) \right)$ is given by
\begin{equation} \label{eq:adjeq}
    \psi' = \begin{pmatrix}
        \frac{(c^\ast)^3 \hat{\tau} v_{f,1}}{1-(c^\ast)^3 \hat{\tau}} \left( 3 v_{f,1} - 2(j_{f,1}+w) \right) & \frac{v_{f,1}}{1-(c^\ast)^3 \hat{\tau}} \left( 3 v_{f,1} - 2(j_{f,1} \right) \\
        \frac{(c^\ast)^3}{1-(c^\ast)^3 \hat{\tau}} \left( 1-\hat{\tau} v_{f,1}^2 \right) & \frac{1}{1-(c^\ast)^3 \hat{\tau}} \left( (c^\ast)^3-v_{f,1}^2 \right)
    \end{pmatrix}   \psi.
\end{equation}
The one-dimensional space of solutions which grow as $\zeta \to \infty$ at most algebraically is spanned by
\begin{equation}
\begin{aligned}
    \psi_{f_1}(\zeta; w) &:= e^{(c^\ast)^3 \zeta(1-2\hat{\tau}w)-(1-3 (c^\ast)^3\hat{\tau}) \int_0^\zeta v_{f,1}^2 , \mathrm{d}\zeta-2(c^\ast)^3 \hat{\tau} \int_0^\zeta j_{f,1} \, \mathrm{d}\zeta} \begin{pmatrix}
        -\dot{j}_{f,1}(\zeta; w) \\
        -\dot{v}_{f,1}(\zeta; w)
    \end{pmatrix} \\
    &= \alpha(\zeta; w) \begin{pmatrix}
        -\dot{j}_{f,1}(\zeta; w) \\
        -\dot{v}_{f,1}(\zeta; w)
    \end{pmatrix}.
\end{aligned}
\end{equation}
Indicating the right-hand side of Eq.~\eqref{eq:laypbs0} by $F$, the Melnikov integral $M_c$ is now defined as
\begin{equation}
\begin{aligned}
    M_c &= \int_{-\infty}^\infty \partial_c F(\phi_{f,1}(\zeta; w)) \cdot \psi_{f,1}(\zeta; w) \, \mathrm{d}\zeta = \\
    &= \frac{3 (c^\ast)^3}{(1-(c^\ast)^3\hat{\tau})^2} \int_{-\infty}^\infty \alpha(\zeta; w) (w-v_{f,1}+j_{f,1})v_{f,1}^2 \left( j_{f,1}-(w-v_{f,1}+j_{f,1})v_{f,1}^2 \hat{\tau} \right) \, \mathrm{d}\zeta.
\end{aligned}
\end{equation}
Since all terms in brackets involved in $M_c$ are strictly positive for any arbitrary $c^\ast(w, \hat{\tau})$, this Melnikov integral always has a definite positive sign. This implies that the distance function $D(\tilde{c}; w, \hat{\tau})$ defined in Eq.~\eqref{eq:distc} is nonconstant to leading order; in particular, it changes linearly with respect to $\tilde{c}$. Therefore, the zero level set of such function corresponds to a smooth curve passing through $c_0$, which in turn proves that the intersection between $\mathcal{W}^{ss}(\mathcal{M}^{(1)})$ and $\mathcal{W}^u(\mathcal{M}^{(2)})$ persists and varies smoothly -- i.e., it is transverse. \\
The last step in order to compute an asymptotic expansion for $c^\ast(\hat{\tau})$ for our travelling pulse (corresponding to $w=a$) is based on a variables rescaling, namely
\begin{equation} \label{eq:rescT}
    Z = \frac{v}{w}, \quad J = \frac{j}{w}, \quad \mu = c^{3/2} w \zeta, \quad \theta = \frac{c^{3/2}}{w}, \quad T = \hat{\tau} w^2.
\end{equation}
Eq.~\eqref{eq:laypbs0} thus becomes
\begin{equation} \label{eq:rescalZJ}
\begin{aligned}
    Z' &= \theta \left( (1+J-Z)T Z^2-J \right), \\
    J' &= \frac{1}{\theta} \left( (1+J-Z)Z^2-\theta^2 J \right),
\end{aligned}
\end{equation}
where we have adopted the notation $' = (1-\theta^2 T)\frac{d}{d\mu}$. In this planar system, the equilibria $p_1(w)$ and $p_2(w)$ correspond to $(0,0)$ and $(1,0)$, respectively. The range in Eq.~\eqref{eq:rangecw} hence corresponds to
\begin{equation} \label{eq:rangetheta}
\frac{1}{3+T} < \theta^2(T) < \frac{1}{2 T} \left( 1-\sqrt{1-32 T} \right),
\end{equation}
where $\theta(T)$ can be approximated numerically as the function such that the distance between the expansion of both $\mathcal{W}^{ss}(0,0)$ and $\mathcal{W}^{u}(1,0)$ up to the sixth order at $Z=\frac12$ is $0$. We note that for $T=0$ we retrieve the value $\theta_0 \approx 0.8615$, consistently with \cite{Carter2018}. Plotting $\theta(T)$ in the region of the $(\theta, T)$-plane where $\theta$ is in a neighbourhood of $\theta_0$ and $T \in \left(0 ,\frac{1}{32} \right)$ suggests to expand $\theta$ as a function of $T$ yielding
\begin{equation} \label{eq:thetaT}
    \theta(T) = \theta_0 + \theta_1 T + \theta_2 T^2 + \mathcal{O}(T^3),
\end{equation}
where $\theta_1 \approx 0.0183$ and $\theta_2 \approx 0.0120$. This expression gives an accuracy of $\theta(T)$ up to $\mathcal{O}(10^{-6})$. Undoing the rescaling in \eqref{eq:rescalZJ} and focusing our attention on $w=a$, we obtain Eq.~\eqref{eq:c_ast}.
\end{proof}

\begin{remark} \label{rem:constrtau}
    Whereas the constraint $\hat{\tau} < \frac{1}{c^3}$ is fundamental in the setup of the analysis of System \eqref{eq:postscaling}, the constraint $\hat{\tau}< \frac{1}{32 a^2}$ in Lemma \ref{lemma:fc2} is the result of fairly conservative estimates in the construction of the overshooting function $\hat{j}(v)$. As such, the statement of Lemma \ref{lemma:fc2} does not rule out the existence of a heteroclinic between $p_1(w)$ and $p_2(w)$ that approaches $p_1(w)$ along its strong stable manifold \eqref{eq:Wssp1w} for values $\hat{\tau} > \frac{1}{32 a^2}$. However, analysis of the saddle $p_2(w)$ reveals that for $\hat{\tau} > \frac{1}{w^2}$, the unstable manifold of $p_2(w)$ lies to the right of the vertical line $\left\{ v = w\right\}$, since $\dot{v} > 0$ on that line. In particular, this prevents the existence of the desired heteroclinic orbit. Hence, we obtain the necessary condition $\hat{\tau} < \frac{1}{w^2} = \frac{1}{a^2}$. Numerical investigations reveal that a heteroclinic orbit exists for $\frac{1}{32 a^2} <\hat{\tau} < \frac{1}{4 a^2}$, where the upper limit can be increased by applying higher precision in the numerical shooting approach.
\end{remark}

\begin{lemma} \label{lemma:fc1}
There exists a unique $s= s^\ast = \frac{a}{\mathcal{D}}$ for which System \eqref{eq:laypb} admits a heteroclinic orbit from $p_3(a,s^\ast) = (a,0,s^\ast,0)$ to $p_2(a) = (a,a,0,0)$. This heteroclinic orbit lies in the hyperplane $\left\{ (w,v,s,j) \, : \, w=a,\,w+j-v-\mathcal{D} s = 0 \right\}$.
\end{lemma}

\begin{proof}
We start by fixing $w=a$ and follow the same approach used in \cite[Lemma 3.2]{grifo2025far}, showing that a heteroclinic connection between $p_3(a,s^\ast)$ and $p_2(a)$ exists on an invariant hyperplane $\Pi$ as follows (further details are available in the proof of \cite[Lemma 3.2]{grifo2025far}). For an arbitrarily fixed $s=\bar{s}$, the point $p_3(a,\bar{s}) \in \mathcal{M}^{(3)}$ admits eigenvalues
\begin{equation}\label{eq:eigenvalues_p3}
    \lambda^c_3(a,\bar{s}) = 0\quad\text{and}\quad \lambda^{s,u}_3(a,\bar{s}) = -\frac{1}{2}c^{3}(1-\hat{\tau} \mathcal{H}\bar{s})\left(1 \mp \sqrt{1 + \frac{4 \mathcal{H} \bar{s}}{c^3(1-\hat{\tau} \mathcal{H}\bar{s})^2}}\right),   
\end{equation}
with associated eigenvectors $\bm{\eta}^c$, $\bm{\eta}^s$, and $\bm{\eta}^u$, respectively. By performing a coordinate change $(v,s,j) = (0,\bar{s},0) + w_c \bm{\eta}^c + w_s \bm{\eta}^s + w_u \bm{\eta}^u$, we obtain for the centre dynamics that the plane $\left\{ w_c=0 \right\}$ is globally invariant if and only if \mbox{$\bar{s} = \frac{a}{\mathcal{D}} =: s^\ast$}. This plane coincides with $\Pi = \left\{a +j - v - \mathcal{D}s = 0\right\}$ in the original coordinates. Both $p_2(a)$ and $p_3(a,s^\ast)$ lie on $\Pi$, and the eigenvalues of $p_2(a)$ are given by
    \begin{equation}\label{eq:eigenvalues_p2}
    \lambda^{(1)}_2(a) = -c^3,\quad \lambda^{(2)}_2(a)= -\frac{\mathcal{H} a}{\mathcal{D}}(1-c^3\hat{\tau}),\quad \lambda^{(3)}_2(a) = a^2(1-c^3\hat{\tau}).
    \end{equation}
The stable manifold $\mathcal{W}^s(p_2(a))$ is two-dimensional and the stable eigenspace $E^s(p_2(a)) \subset \Pi$, which implies that $E^s \equiv \Pi$. Moreover, $\mathcal{W}^s(p_2(a))$ is tangent to $\Pi$ at $p_2(a)$ (hyperbolic) and both $\mathcal{W}^s(p_2(a))$ and $\Pi$ are two-dimensional invariant manifolds under the flow \eqref{eq:laypb}. Because of the local uniqueness of the stable manifold, they must coincide. \\
On the hyperplane $\Pi$ Eq.~\eqref{eq:laypb} becomes
\begin{equation} \label{eq:laypbPi}
    \begin{aligned}
    \dot{v} &= \frac{c^3}{1-c^3\hat{\tau}} \left( - \frac{\hat{\tau}\mathcal{H}}{\mathcal{D}}(a+j-v)v -j \right), \\
    \dot{j} &= \frac{1}{1-c^3\hat{\tau}} \left( - \frac{\mathcal{H}}{\mathcal{D}} (a+j-v) v - c^3 j \right),
    \end{aligned}
\end{equation}
where the steady-state $(v,j)=(0,0)$ (saddle) corresponds to $p_3(a,s^\ast)$. By construction, its unstable manifold lies in $\mathcal{W}^s(p_2(a))$, thus yielding the existence of a heteroclinic connection between $p_3(a,s^\ast)$ and $p_2(a)$ as intended. 
\end{proof}

\subsubsection{Singular orbit} \label{sec:singorb}
We can now construct a full homoclinic orbit to $p_1(a)$ by matching orbits from the reduced problems obtained in Sec.~\ref{sec:redprob} with the heteroclinic connections showed in Lemmas \ref{lemma:fc2}-\ref{lemma:fc1}. To this aim, we introduce the following notation for subsets of $\mathcal{M}^{(1)}$ and $\mathcal{M}^{(3),w_0}$ for any $\underline{w}$, $\overline{w}$ such that $\underline{w}<\overline{w}$:
\begin{equation}
\begin{aligned}
 \mathcal{M}^{(i)} (\underline{w},\overline{w}) &= \left\{ p_i(w) \, : \, w \in [\underline{w},\overline{w}] \right\}, \qquad i=1,2, \\
 \mathcal{M}^{(3),w_0} (\underline{w},\overline{w}) &= \left\{ p_3(w,s(w)) \, : \, w \in [\underline{w},\overline{w}] \right\}.
\label{eq:M1M3w}
\end{aligned}
\end{equation}
We then obtain the following result.
\begin{proposition}\label{thm:singsol}
 For any fixed value of $a$, $\mathcal{B}$, $\mathcal{D}$, $\mathcal{H}>0$, $0 < \hat{\tau} < \mathrm{min} \left\{ \frac{1}{c^3}, \frac{1}{a^2} \right\}$ and given $c(\hat{\tau})=c^\ast(\hat{\tau})=(a \, \theta (a,\hat{\tau}))^{2/3}$, there exists a unique singular orbit $\Phi_0(\hat{\tau})$ homoclinic to $p_1(a)=(a,0,0,0)$ corresponding to a singular solution to System \eqref{eq:postscalingw}, composed of segments
of orbits of the layer problem \eqref{eq:laypb} and the reduced problems \eqref{eq:slowflow_M3} and \eqref{eq:superslowM1}. In particular, $\Phi_0(\hat{\tau})$ is given by the concatenation
\begin{equation} \label{eq:singhom}
  \Phi_0(\hat{\tau})=\mathcal{M}^{(1)}(0, a) \cup \mathcal{M}^{(3),0}(0, a) \cup \phi_{f,2} \cup \phi_{f,1},
\end{equation}
where $\phi_{f,1}$ and $\phi_{f,2}$ are the heteroclinic orbits determined in Lemma \ref{lemma:fc2} and Lemma \ref{lemma:fc1}, respectively.
\end{proposition}
\begin{proof}
It follows the same steps as the proof of \cite[Prop.~3.3]{grifo2025far}.
\end{proof}

\subsubsection{Persistence} \label{sec:persist}

In this section, we first prove that the singular orbit $\Phi_0(\hat{\tau})$ obtained in Prop.~\ref{thm:singsol} perturbs to a solution $\Phi_\delta(\hat{\tau})$ of Eq.~\eqref{eq:postscalingw} for $0 < \delta \ll 1$. Then, we finalise the proof of Thm.~\ref{thm:main}, in particular showing the leading order expression for the pulse speed $\mathcal{C}^\ast$.

\begin{theorem} \label{thm:pertsol} 
Let $a$, $\mathcal{B}$, $\mathcal{D}$, $\mathcal{H}$, $\hat{\tau}$ be fixed and positive, and let $\hat{\tau} < \mathrm{min} \left\{ \frac{1}{c^3}, \frac{1}{a^2} \right\}$. Let $0<\delta \ll 1$ be sufficiently small, and assume that $a$, $\mathcal{B}$, $\mathcal{D}$, $\mathcal{H}$ are $\mathcal{O}(1)$ in $\delta$. There exists a unique value $c(\hat{\tau})=c^\ast_\delta(\hat{\tau})$ for which a unique orbit $\Phi_\delta(\hat{\tau})$ to \eqref{eq:postscalingw} exists that is homoclinic to $p_1(a) = (a,0,0,0)$. Moreover, $\Phi_\delta(\hat{\tau})$ is $\mathcal{O}(\delta)$ close to $\Phi_0(\hat{\tau})$ \eqref{eq:singhom}, and $c^\ast_\delta(\hat{\tau})=c^\ast(\hat{\tau})+\mathcal{O}(\delta)$ \eqref{eq:c_ast}.
\end{theorem}

\begin{proof}
The proof is analogous to \cite[Thm.~3.4]{grifo2025far}. In particular, we observe that $\mathcal{M}^{(3)}$ is an invariant manifold for the full system \eqref{eq:postscalingw}; on this manifold, the dynamics (independent of $\hat{\tau}$) are again explicitly given by $w(s)=w_3(s; k_1)$. The $\alpha$-limit set of $\Sigma$ -- a neighbourhood of $p_3(a, s^\ast)$ on $\mathcal{M}^{(3)}$ -- solely consists of $p_1(a)$. Moreover, the double heteroclinic orbit $\phi_{f,2} \cup \phi_{f,1}$ constructed in Prop.~\ref{thm:singsol} is valid for arbitrary values of $w$, yielding $c_w^\ast(\hat{\tau})=(w \, \theta(w, \hat{\tau}))^{2/3}$ and $s_w^\ast = \frac{w}{\mathcal{D}}$. Tracking backwards the (fast) flow of a neighbourhood of $p_1(a)$ given by $\mathcal{M}^{(1)}(a-\mu \mathcal{D}, a+\mu \mathcal{D})$ (with $\mu$ arbitrarily small) until it reaches $\mathcal{M}^{(3)}$ again, we obtain a transversal intersection; this corresponds to a curve $\Gamma$ such that $\Gamma \cap \Sigma = p_3(a, s^\ast)$ for $\delta=0$, and, thanks to Fenichel theory and the Exchange Lemma, a curve $\Gamma_\delta$ such that $\Gamma_\delta \cap \Sigma = p_3(a, s_\delta^\ast)$ for $0 < \delta \ll 1$, where $p_3(a, s_\delta^\ast)$ is $\mathcal{O}(\delta)$-close to $p_3(a, s^\ast)$. The orbit is finally closed via $w_3(s; k_1^\ast)$ on $\mathcal{M}^{(3)}$ until $w=a$, where we note that $k_1^\ast$ is the constant obtained from solving $w_3(s^\ast_\delta; k_1^\ast) = a$.
\end{proof}

We now turn our attention to the proof of Thm.~\ref{thm:main}.


\begin{proof}[Proof of Thm.~\ref{thm:main}]
By undoing the rescalings that led to Eq.~\eqref{eq:postscalingw} and \eqref{eq:thetaT}, we are here going to provide leading order expressions for the critical speed $\mathcal{C}^\ast$ and for the upper bound $\tau_{\rm upper}$ limiting the range of admissible $\tau$-values covered by our analysis of far-from-equilibrium travelling pulses. To this aim, we first reformulate Eq.~\eqref{eq:thetaT} to leading order as $\theta(T) = \sum_i \theta_i T^i$, where the values of $\theta_i$, $i=0,1,2$ have been derived in the proof of Lemma \ref{lemma:fc2}. According to Eq.~\eqref{eq:rescT}, this is equivalent to $c^{3/2} = \sum_i \theta_i \hat{\tau}^i w^{2i+1}$, which, reverting to the original scaling via Eq.~\eqref{eq:rescaling} and recalling that in our case $w=a$, yields 
\begin{equation} \label{eq:CepsFV}
(\mathcal{C} \varepsilon)^{3/2} = \varepsilon \sum_i \theta_i \left(\frac{\tau}{\mathcal{C} \varepsilon}\right)^i \left[(1+\mathcal{C} \varepsilon)\mathcal{A}\right]^{2i+1}. 
\end{equation}
This represents an implicit equation for $\mathcal{C}$, which can be solved numerically for any given fixed value of $\mathcal{A}$ and $\varepsilon$. To leading order, we obtain
\begin{equation}
    \mathcal{C}^\ast \varepsilon= (\mathcal{A} \theta_0)^{2/3} \varepsilon^{2/3} + \mathcal{O}(\varepsilon^{4/3}) + \left[\frac{2 \mathcal{A}^2 \theta_1}{3 \theta_0} + \frac{16 \mathcal{A} \theta_1}{9 \mathcal{A}^{1/3} \theta_0^{1/3}}\varepsilon^{2/3} + \mathcal{O}(\varepsilon^{4/3})\right] \tau + \mathcal{O}(\tau^2).
\end{equation}
In particular, applying the rescalings \eqref{eq:rescaling}-\eqref{eq:rescT} and recalling that to leading order $a \approx \mathcal{A}$, we obtain
\begin{equation} \label{eq:Casttau}
    \mathcal{C}^*(\tau) =\left( \frac{\mathcal{A}^2 \theta(\tau)^2}{\varepsilon} \right)^{1/3} + \mathcal{O}(1)
\end{equation}
as claimed. 
In Fig.~\ref{fig:bifdiag}(a) we provide a comparison between $\mathcal{C}^\ast$ as a function of $\tau$ and the numerical data obtained from AUTO continuation (see Sec.~\ref{sec:numcont}) for a fixed value of $\mathcal{A}$ and $\varepsilon$.

\noindent
As for the upper bound $\tau_{\rm upper}$, we recall that the constraint $\hat{\tau} < \min \left\{ \frac{1}{c^3}, \frac{1}{a^2} \right\}$ must be satisfied for travelling pulses to exist in this regime. Reverting the condition $\hat{\tau} < \frac{1}{c^3}$ to the original scaling via Eq.~\eqref{eq:rescaling} yields $\tau < \frac{1}{\mathcal{C}^2}$, which in virtue of Eq.~\eqref{eq:Casttau} leads to $\tau < \varepsilon^{2/3} \mathcal{A}^{-4/3} \theta(\tau)^{-4/3} + \mathcal{O}(\varepsilon)$. On the other hand, reverting the condition $ \hat{\tau} < \frac{1}{a^2}$ to the original scaling, we find $\tau < \frac{\varepsilon \mathcal{C}}{\mathcal{A}^2 (1+\varepsilon \mathcal{C})^2}$, which again by Eq.~\eqref{eq:Casttau} yields $\tau < \varepsilon^{2/3} \mathcal{A}^{-4/3}\theta(\tau)^{2/3} + \mathcal{O}(\varepsilon)$. We observe that this is internally consistent: by Eq.~\eqref{eq:rescaling}, $\tau$ should be of order $\varepsilon^{2/3}$; in particular, $\tau$ is (reasonably) small. Therefore, we can approximate $\theta(\tau) \approx \theta_0$ to leading order for small $\varepsilon$. Since $\theta_0^{2/3}<\theta_0^{-4/3}$, the second condition on $\tau$ prevails on the first, and we obtain to leading order
\begin{equation} \label{eq:taumax}
    \tau_{\rm upper} = \varepsilon^{2/3} \mathcal{A}^{-4/3}\theta_0^{2/3} + \mathcal{O}(\varepsilon).
\end{equation}
\end{proof}

\begin{remark} \label{rem:smalltau}
    As shown in the proof of Thm.~\ref{thm:main}, the assumption on the smallness of $\tau$ (in particular $\tau = \mathcal{O}(\varepsilon^{2/3})$ with $0 < \varepsilon \ll 1$) introduced in Eq.~\eqref{eq:rescaling} is crucial for the construction of the travelling pulses. In particular, assuming both $\tau$ and $\mathcal{C}$ to be order $1$ in System \eqref{eq:prescaling} would prevent us from finding a suitable scaling leading to System \eqref{eq:postscaling}.
\end{remark}

\subsection{Numerical continuation} \label{sec:numcont}

To validate the asymptotic analysis and explore the role of the inertial parameter $\tau$ on far-from-equilibrium pulse dynamics, we perform a numerical continuation of System \eqref{eq:prescaling} using the software \textsc{AUTO} \cite{Doedel1981}. The bifurcation diagrams reported in Fig.~\ref{fig:bifdiag} summarise the dependence of the wave speed and the state variables on $\tau$, starting from the reference configuration used in Fig.~\ref{fig:num_GSPT}, i.e., $\mathcal{A} = 1.2$, $\mathcal{B} = 0.45$, $\varepsilon = 0.005$, $\mathcal{D} = 4.5$, $\mathcal{H} = 1$. In particular, panel (a) shows that the migration speed $\mathcal{C}$ increases monotonically as $\tau$ grows, leading to faster uphill propagation of the vegetation pulse. We observe that this behaviour is observed even for $\tau$-values beyond the analytical upper bound, which according to Eq.~\eqref{eq:taumax} here corresponds to $\tau_{\rm upper} \approx 0.0208$. In general, this result is in perfect agreement with the theoretical predictions and provides a strong validation of the asymptotic analysis presented in the previous sections. From an ecological perspective, although inertia generally introduces a delayed response in the system, especially in the transient regime between asymptotic states (e.g. from a steady-state to a patterned state or between different patterned states),
this effect does not prevent the increase in migration speed. The reason apparently lies in the fact that the inertial term seems to act as a transport enhancer rather than a growth facilitator: while local dynamics become slower, the advective component dominates the overall behaviour, resulting in faster uphill propagation of vegetation bands despite reduced local adaptation. 

Moreover, panels (b)-(d) illustrate the variation of the pulse peak w.r.t.~its state variables, namely $U_{max}$, $V_{max}$, and $S_{max}$, as the inertial time increases. The biomass amplitude $V_{max}$ exhibits a non-monotonic trend: it initially increases for small values of $\tau$, reaches a maximum, and then decreases as $\tau$ becomes large. In contrast, the toxicity component $S_{max}$ decreases monotonically with $\tau$, indicating that stronger inertial effects reduce the accumulation of autotoxic compounds. The amplitude of the surface water $U_{max}$ remains essentially constant, showing negligible variations compared to the other state variables. The change in monotonicity for $V_{max}$ occurs approximately close to $\tau_{\rm upper}$: solutions on the part of the branch beyond this value, in fact, do not present the same structure as the one of the travelling pulses investigated here, and cannot be covered by the above analysis. From an ecological perspective, this behaviour suggests that moderate inertial effects may temporarily favour biomass growth by stabilising resource redistribution, while large $\tau$-values lead to a decline in vegetation due to the reduced time available for water uptake during rapid migration. The monotonic decrease of $S_{max}$ with $\tau$ implies that faster-moving pulses leave less time for autotoxic compounds to accumulate locally, partially mitigating toxicity stress. However, this advantage is accompanied by a qualitative change in biomass distribution: while $V_{max}$ initially increases with $\tau$, the faster migration associated with large $\tau$-values reduces the residence time of vegetation in each location, potentially limiting long-term resource exploitation despite lower toxicity levels. 

\begin{figure}[t!]
	\centering
	\includegraphics[width=1\textwidth]{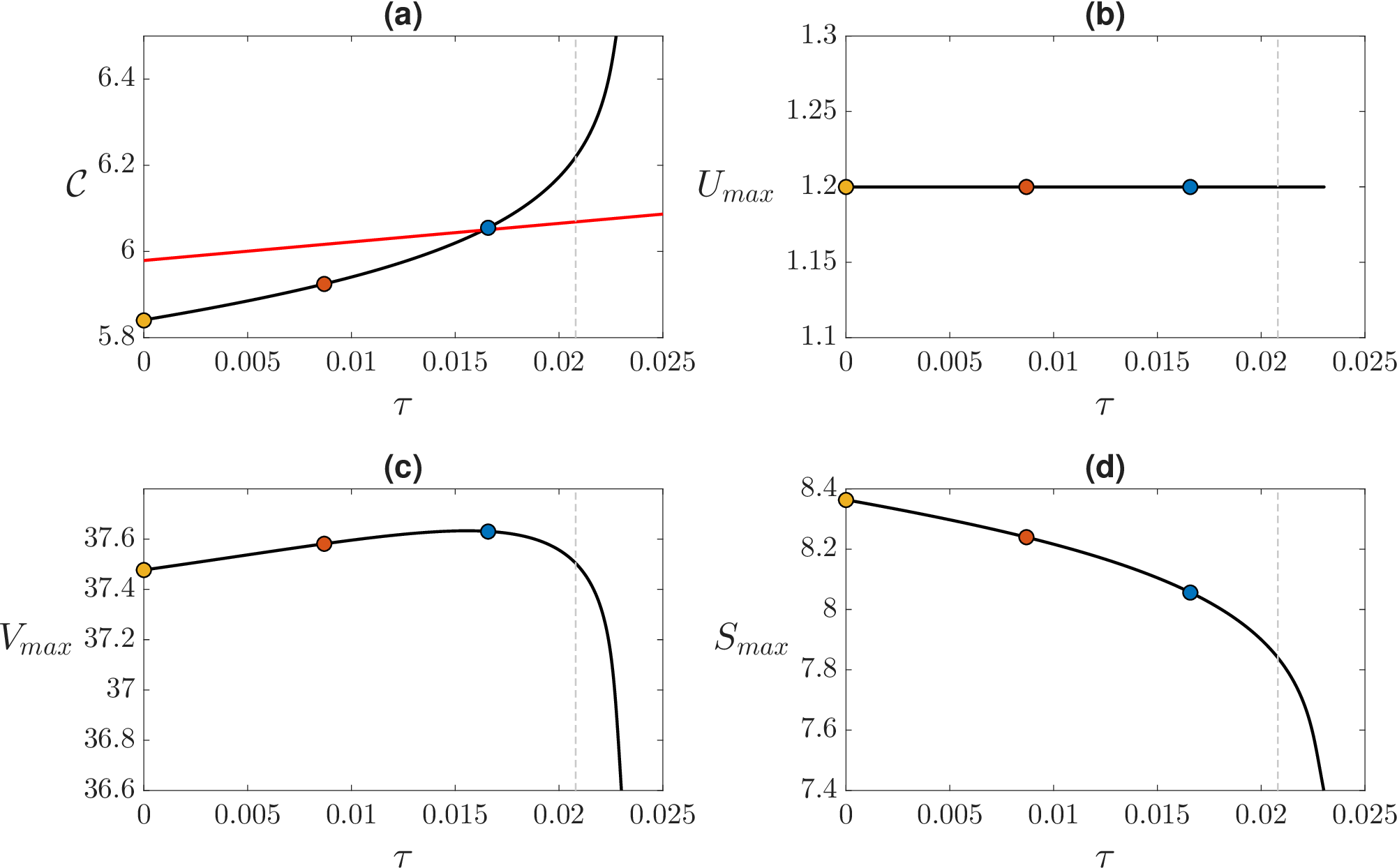}
        \caption{Bifurcation diagrams obtained by performing numerical continuation of System \eqref{eq:prescaling} with AUTO varying $\tau$. Panel (a) depicts how the migration speed $\mathcal{C}$ is affected by the inertial time: the black solid line is obtained from numerical continuation, whereas the red solid line represents the migration speed derived analytically in Thm.~\ref{thm:main}. Panels (b-d) show the dependence of the pulse peak expressed in terms of the field variables $(U_{max},V_{max},S_{max})$ on $\tau$. In all panels, the dashed grey line indicates the location of the threshold $\tau_{\rm upper}$ reported in \eqref{eq:taumax}. The starting point for all the figures is the same as the one fixed in Fig.~\ref{fig:num_GSPT} ($\mathcal{A}=1.2$, $\mathcal{B}=0.45$, $\varepsilon=0.005$, $\mathcal{D}=4.5$, $\mathcal{H}=1$ and $\tau = 0$). The coloured dots on each branch identify the location in the corresponding plane of the solutions illustrated in Fig.~\ref{fig:profiles}.}
        \label{fig:bifdiag}
\end{figure}

Further insight into the qualitative changes of the pulse solution is provided in Fig.~\ref{fig:profiles}, where profiles of $U$, $V$, $J$, and $S$ are plotted in the comoving frame $\xi$ for different $\tau$-values. As $\tau$ grows, the vegetation peak becomes narrower while the toxicity profile broadens. The resulting asymmetry in the pulse shape suggests that inertial effects hinder uphill colonisation, reinforcing the dominance of negative plant-soil feedback. Note that these morphological changes are ecologically relevant as they reproduce the transition from resilient vegetation bands to degraded states under chronic stress.

\begin{figure}[t!]
	\centering
	\includegraphics[width=0.9\textwidth]{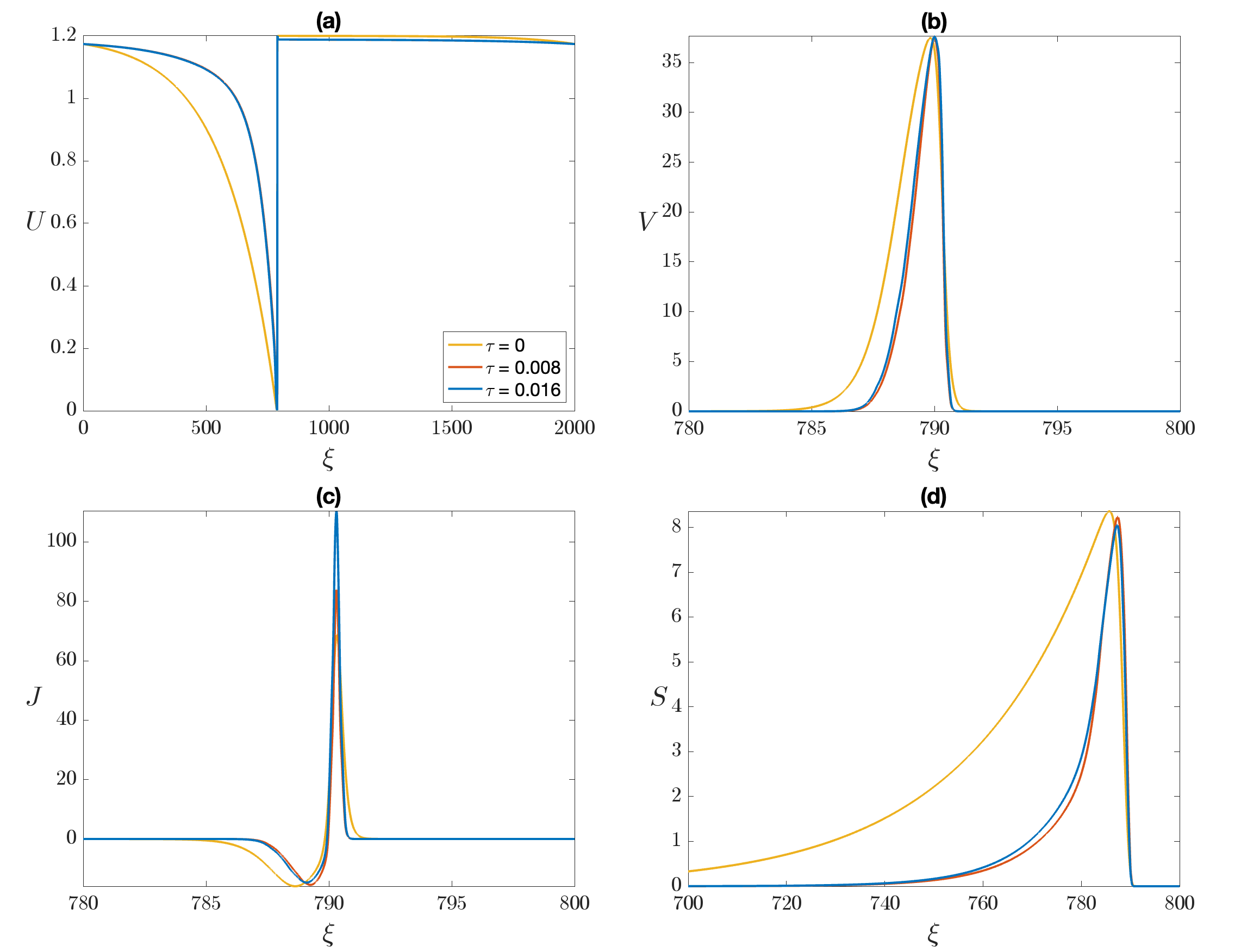}
        \caption{Profiles of the field variables (panel (a): $U$, panel (b): $V$, panel (c): $J$, panel (d) $S$) in terms of the comoving variable $\xi$, for different values of $\tau$ obtained by performing numerical continuation of System \eqref{eq:prescaling} with AUTO. Different colours correspond to different parameter configurations as shown in Fig.~\ref{fig:bifdiag}.
        }
        \label{fig:profiles}
\end{figure}

In order to compare the numerical solutions depicted in Fig.~\ref{fig:bifdiag}-\ref{fig:profiles} with the analytical results obtained in the asymptotic case $0<\delta\ll 1$, we plot in Fig.~\ref{fig:orbits} these profiles as homoclinic orbits projected into $(v,s,w)$-space (panel (a)) and $(v,s,q)$-space (panel (b)), with colours corresponding to the solutions shown in Fig.~\ref{fig:bifdiag}-\ref{fig:profiles}. It should be noticed that the shape of the homoclinic orbit remains essentially unchanged across the considered range of $\tau$, confirming that inertia does not alter the qualitative structure of the connection. This is in perfect agreement with the theoretical predictions: the inertial term influences the migration speed and the amplitude of the state variables, but it does affect neither the geometry of the singular skeleton described in the previous sections nor the different scales. In particular, the equilibria $p_{2}(a)$ and $p_{3}(a,s^{\ast})$ preserve their relative positions in phase space, and the global organisation of the orbit remains consistent with the analytical construction based on the layer problem. These findings validate the asymptotic analysis and highlight the robustness of the theoretical framework.

\begin{figure}[t!]
	\centering
    \psfrag{PLACEHOLDER}{$\mathscr{C}$}
	\includegraphics[width=1\textwidth]{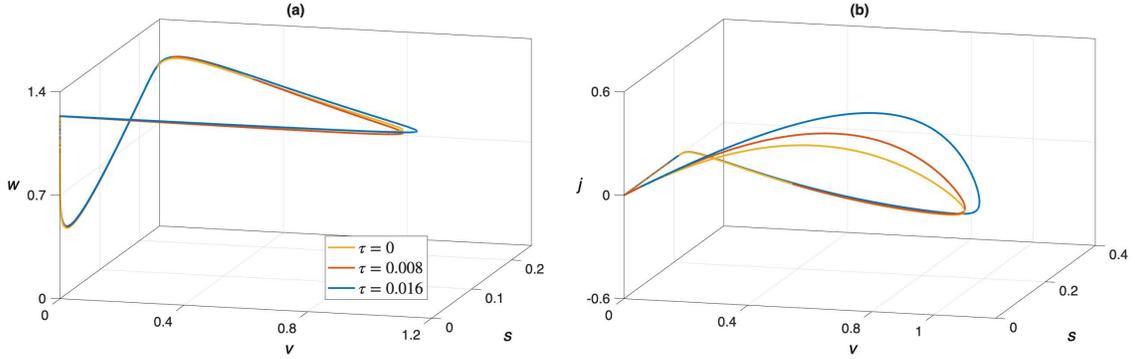}
        \caption{Homoclinic orbits obtained by performing numerical continuation of System \eqref{eq:prescaling} with AUTO in ($v,s,w$)-space (panels (a)) and in ($v,s,q$)-space (panels (b)). Different colours correspond to different parameter configurations as shown in Fig.~\ref{fig:bifdiag}.}
        \label{fig:orbits}
\end{figure}

In summary, the numerical continuation confirms the theoretical predictions derived from the asymptotic analysis. The migration speed $\mathcal{C}$ increases monotonically with $\tau$ up to the threshold $\tau = \tau_{\max}$, validating the analytical framework and highlighting the role of inertia as a transport enhancer. While the inertial term significantly affects the propagation speed and the amplitude of the state variables, it does not alter the qualitative structure of the pulse solution: the homoclinic orbit and the underlying singular skeleton remain unchanged across the considered range of $\tau$. These findings demonstrate the robustness of the geometric construction and provide a consistent link between the analytical and numerical approaches for far-from-equilibrium dynamics.

\section{Discussion}\label{sec:Conclusion}



The analysis developed in this paper provides a framework for understanding how inertial effects modulate vegetation patterns in sloped semi-arid landscapes, highlighting the differences that emerge between close to the instability threshold and far-from-equilibrium dynamics. By combining LSA, WNA, and GSPT, we have shown that inertia does not simply introduce a temporal lag in the vegetation response, but rather acts as a genuinely dynamical parameter that reshapes both the onset and the nonlinear evolution of coherent structures. 

Comparing the traveling patterns emerging in the analysed regimes, some important differences are observed both analytically and numerically. On the one hand, patterned solutions exhibit a larger sensitivity to inertial effects in the far-from-equilibrium regime. Indeed, while the wave bifurcation locus is affected by inertia for $\tau > 0.1$ only, the traveling pulses features respond to any $\tau \not = 0$. On the other hand, the migration speed shows opposite inertial dependencies in the two regimes. Specifically, it scales inversely with inertial time near the onset but directly in the far-from-equilibrium regime (see Figs. \ref{figure6} and \ref{fig:bifdiag}). This apparent contradiction stems from the fact that, in the first scenario, the speed reduction is mainly dictated by non-negligible reduced plant losses.

As for the travelling pulses analysed in Sec.~\ref{sec:travpulse}, the range of $\tau$-values supporting their existence, as further shown in Sec.~\ref{sec:numcont}, is fairly narrow, especially when compared to the one supporting the emergence of small-amplitude periodic patterns close to the onset of instability. This is due to the fact that, as pointed out in Rem.~\ref{rem:smalltau}, the regime covered by our GSPT investigation requires $\tau$ to be $\mathcal{O}(\varepsilon^{2/3})$ with $0 < \varepsilon \ll 1$. The case $\tau = \mathcal{O}(1)$, analysed in Sec.~\ref{sec:closeto}, is therefore excluded from the analysis carried out in Sec.~\ref{sec:travpulse}. Moreover, we observe that inertia does not influence the qualitative structure of such travelling pulses, but rather modulates their quantitative features, such as propagation speed and amplitude. In particular, since the speed increases and the pulses become narrower as $\tau$ increases, inducing the use of refined numerical methods to capture the pulses' structure correctly. In the hyperbolic regime, as long as $\tau$ increases, the pulse becomes extremely sharp, requiring very fine spatial meshes to avoid numerical artifacts. Stable simulations require a mesh size small enough to resolve the peak width and an implicit time–stepping scheme to handle the stiffness. Coarser grids or large time steps lead to an underestimation of both amplitude and propagation speed.

In the future, we plan to construct and analyse the associated Busse baloon \cite{EIGENTLER2024,VanDerStelt2013} to link the two regimes analysed in this work. Moreover, we aim to consider the more general scenario where multiple species (with specific, distinguished relaxation times) interact. From a more ecological perspective, a further research direction involves the introduction of parameters dependent on space and/or time in order to take into account, for instance, soil topography and seasonal rainfall, which play a particularly important role in arid and semi-arid environments.

\paragraph{Acknowledgements.}
GC, CC, GG, AI and GV are members of Gruppo Nazionale per la Fisica Matematica (GNFM), Istituto Nazionale di Alta Matematica (INdAM), which has partially supported this work. GG is a post-doc fellow supported by the National Institute of Advanced Mathematics (INdAM) funded by the European Union - NEXTGENERATIONEU CUP E63C25000470007. AI acknowledges the financial support from the project ``Pluvial Adaptation and Resilience scenarios for urban Areas through improved Data of Intense rainfall, Good practices, and water Management - PARADIGM'' (Grant No. MUR PE00000005 - CUP E63C22002000002) as well as ``Sustainable Urban areas by Nature-based solutions implementation to mitigate climate impacts and achieve a Resilient, Innovative and Smart Environment - SUNRISE'' (Grant No. MUR PE00000005 - CUP E63C22002000002). GC, CC and GV acknowledges support from MUR (Italian Ministry of University and Research) through PRIN2022-PNRR Project No.~P2022WC2ZZ ``A multidisciplinary approach to evaluate ecosystems resilience under climate change''.

\bibliographystyle{unsrtnat}
\bibliography{bibliography}

@article{Carter2018,
  title		=	{Traveling stripes in the {Klausmeier} model of vegetation pattern formation},
  author	=	{Carter, P. and Doelman, A.},
  journal	=	{SIAM Journal of Applied Mathematics},
  volume	=	{78},
  pages		=	{3213--3237},
  year		=	{2018},
  DOI 		= 	{10.1137/18M1196996}
}

@article{Marasco2014,
  title		=	{Vegetation pattern formation due to interactions between water availability and toxicity in plant-soil feedback},
  author	=	{Marasco, A. and Iuorio, A. and Carten\`i, F. and Bonanomi, G. and Tartakovsky, D.M. and Mazzoleni, S. and Giannino, F.},
  journal	=	{Bulletin of Mathematical Biology},
  volume	=	{76},
  pages		=	{2866--2883},
  year		=	{2014},
  DOI 		=	{10.1007/s11538-014-0036-6}
}

@article{Iuorio2021,
  title		=	{The influence of autotoxicity on the dynamics of vegetation spots},
  author	=	{Iuorio, A. and Veerman, F.},
  journal	=	{Physica D},
  volume	=	{427},
  pages		=	{133015},
  year		=	{2021},
  DOI 		= 	{10.1016/j.physd.2021.133015}
}

@article{Klausmeier1999,
  title		=	{Regular and irregular patterns in semiarid vegetation},
  author	=	{Klausmeier, C.A.},
  journal	=	{Science},
  volume	=	{284},
  pages		=	{1826--1828},
  year		=	{1999},
  DOI 		= 	{10.1126/science.284.5421.1826}
}

@article{Bastiaansen2019,
  title		=	{Stable planar vegetation stripe patterns on sloped terrain in dryland ecosystems},
  author	=	{Bastiaansen, R. and Carter, P. and Doelman, A.},
  journal	=	{Nonlinearity},
  volume	=	{32},
  number	=	{8},
  pages		=	{2759},
  year		=	{2019},
  DOI 		= 	{10.1088/1361-6544/ab1767}
}

@article{Doedel1981,
  title={{AUTO: A program for the automatic bifurcation analysis of autonomous systems}},
  author={Doedel, E.J.},
  journal={Congressus Numererantium},
  volume={30},
  number={265-284},
  pages={25--93},
  year={1981}
}

@book{Tongway2001,
  author 	=	{Tongway, D.J.},
  year 		= 	{2001},
  title 	= 	{Banded Vegetation Patterning in Arid and Semiarid Environments},
  publisher = 	{Springer},
  address 	= 	{New York},
  edition 	= 	{first},
  DOI 		= 	{10.1007/978-1-4613-0207-0}
}

@article{Rietkerk2021,
  title		=	{Evasion of tipping in complex systems through spatial pattern formation},
  author	=	{Rietkerk, M. and Bastiaansen, R. and Banerjee, S. and Van De Koppel, J. and Baudena, M. and Doelman, A.},
  journal	=	{Science},
  volume	=	{374},
  pages		=	{},
  year		=	{2021},
  DOI 		= 	{10.1126/science.abj0359}
}

@article{Bastiaansen2020,
  title		=	{The effect of climate change on the resilience of ecosystems with adaptive spatial pattern formation},
  author	=	{Bastiaansen, R. and Doelman, A. and Eppinga, M.B. and Rietkerk, M.},
  journal	=	{Ecology Letters},
  volume	=	{43},
  pages		=	{414--429},
  year		=	{2020},
  DOI 		= 	{10.1111/ele.13449}
}

@article{Consolo2022PRE,
  title		=	{Oscillatory periodic pattern dynamics in hyperbolic reaction-advection-diffusion models},
  author	=	{Consolo, G. and Curr\`o, C. and Grif\`o, G. and Valenti, G.},
  journal	=	{Physical Review E},
  volume	=	{105},
  pages		=	{034206},
  year		=	{2022},
  DOI 		= 	{10.1103/PhysRevE.105.034206}
}

@article{Consolo2022III,
  title		=	{Dryland vegetation pattern dynamics driven by inertial effects and secondary seed dispersal},
  author	=	{Consolo, G. and Grif\`o, G. and Valenti, G.},
  journal	=	{Ecological Modelling},
  volume	=	{474},
  pages		=	{110171},
  year		=	{2022},
  DOI 		= 	{10.1016/j.ecolmodel.2022.110171}
}

@article{Iuorio2023pre,
  title		=	{How does negative plant-soil feedback across lifestages affect the spatial patterning of trees?},
  author	=	{Iuorio, A. and Eppinga, M.B. and Baudena, M. and Veerman, F. and Rietkerk, M. and Giannino, F.},
  journal	=	{Scientific Reports},
  year		=	{2023},
  DOI 		= 	{10.1038/s41598-023-44867-0}
}

@article{VanDerStelt2013,
  title		=	{Rise and Fall of Periodic Patterns for a Generalized {K}lausmeier-{G}ray-{S}cott Model},
  author	=	{{van der Stelt}, S. and Doelman, A. and Hek, G. and Rademacher, J.D.M.},
  journal	=	{Journal of Nonlinear Science},
  volume	=	{23},
  number	=	{7},
  pages		=	{39--95},
  year		=	{2013},
  DOI 		= 	{10.1007/s00332-012-9139-0}
}

@article{Zelnik2013,
  title		=	{Regime shifts in models of dryland vegetation},
  author	=	{Zelnik, Y. and Kinast, S. and Yizhaq, H. and Bel, G. and Meron, E.},
  journal	=	{Philosophical Transactions of the Royal Society A},
  volume	=	{371},
  number	=	{2004},
  pages		=	{20120358},
  year		=	{2013},
  DOI 		= 	{10.1098/rsta.2012.0358}
}

@article{Byrnes2023,
  title={Large amplitude radially symmetric spots and gaps in a dryland ecosystem model},
  author={Byrnes, E. and Carter, P. and Doelman, A. and Liu, L.},
  journal={Journal of Nonlinear Science},
  volume={33},
  number={6},
  pages={107},
  year={2023},
  publisher={Springer},
  DOI = {10.1007/s00332-023-09963-5}
}

@article{Hillerislambers2001,
  title		=	{Vegetation pattern formation in semi-arid grazing systems},
  author	=	{HilleRisLambers, R. and Rietkerk, M. and {van de Bosch}, F. and Prins, H.H.T. and {de Kroon}, H.},
  journal	=	{Ecology},
  volume	=	{82},
  number	=	{1},
  pages		=	{50},
  year		=	{2001},
  DOI 		=	{10.1890/0012-9658(2001)082[0050:VPFISA]2.0.CO;2}
}

@article{Siteur2014,
  title		=	{{Beyond Turing: The response of patterned ecosystems to environmental change}},
  author	=	{Siteur, K. and Siero, E. and Eppinga, M.B. and Rademacher, J.D.M. and Doelman, A. and Rietkerk, M.},
  journal	=	{Ecological Complexity},
  volume	=	{20},
  pages		=	{81--96},
  year		=	{2014},
  DOI 		= 	{10.1016/j.ecocom.2014.09.002}
}

@article{Eigentler2020,
  title		=	{An integrodifference model for vegetation patterns in semi-arid environments with seasonality},
  author	=	{Eigentler, L. and Sherratt, J.A.},
  journal	=	{Journal of Mathematical Biology},
  volume	=	{81},
  pages		=	{875--904},
  year		=	{2020},
  DOI 		=	{10.1007/s00285-020-01530-w}
}

@book{Kuehn_2015,
	doi = {10.1007/978-3-319-12316-5},
	year = 2015,
	publisher = {Springer International Publishing},
	author = {Kuehn, C.},
	title = {{Multiple Time Scale Dynamics}}
}

@article{JKK96,
  title={Tracking invariant manifolds up to exponentially small errors},
  author={Jones, C.K.R.T. and Kaper, T.J. and Kopell, N.},
  journal={SIAM J. Math. Anal.},
  volume={27},
  number={2},
  pages={558--577},
  year={1996},
  publisher={SIAM},
  doi = {10.1137/s003614109325966x}
}

@article {Fe79,
    AUTHOR = {Fenichel, N.},
     TITLE = {Geometric singular perturbation theory for ordinary
              differential equations},
   JOURNAL = {J. Differential Equations},
    VOLUME = {31},
      YEAR = {1979},
    NUMBER = {1},
     PAGES = {53--98},
      ISSN = {0022-0396},
     CODEN = {JDEQAK},
   MRCLASS = {58F30 (34C29)},
  MRNUMBER = {524817 (80m:58032)},
MRREVIEWER = {F. Verhulst},
       DOI = {10.1016/0022-0396(79)90152-9}
}

@article{Consolo2024,
  title		=	{Modeling vegetation patterning on sloped terrains: The role of toxic compounds},
  author	=	{Consolo, G. and Grif\`o, G. and Valenti, G.},
  journal = {Physica D},
  volume = {459},
  number = {},
  pages = {134020},
  year = {2024},
  doi = {10.1016/j.physd.2023.134020}
}

@article{Sewalt2017,
    title   =   {Spatially periodic multipulse patterns in a generalized {K}lausmeier–{G}ray–{S}cott model},
    author  =   {Sewalt, L. and Doelman, A.},
    journal =   {SIAM Journal on Applied Dynamical Systems},
    volume  =   {16},
    number   =   {2},
    pages   =   {1113--1163},
    year    =   {2017},
    doi     =   {10.1137/16M1078756}
}

@book{UN22,
    title = {The {G}lobal {L}and {O}utlook},
    author = {{United Nations Convention to Combat Desertfication}},
    year = {2022},
    url = {https://www.unccd.int/resources/global-land-outlook/global-land-outlook-2nd-edition},
    isbn = {978-92-95118-53-9},
    edition =   {second},
    publisher   =   {UNCCD},
    address     =   {Bonn}
}

@article{grifo2025far,
author = {Grif\`{o}, G. and Iuorio, A. and Veerman, F.},
title = {Far-from-Equilibrium Traveling Pulses in Sloped Semiarid Environments Driven by Autotoxicity Effects},
journal = {SIAM Journal on Applied Mathematics},
volume = {85},
number = {1},
pages = {188-209},
year = {2025},
doi = {10.1137/24M1669499},
}

@software{Comsol,
 author 	=	{{COMSOL Multiphysics $\textsuperscript{\textregistered}$}},
  title 	=	{Ver 6.3 {COMSOL AB}, {S}tockholm, {S}weden},
  url 		=	{https://www.comsol.com/},
  version	= 	{ver 6.3 {COMSOL AB}},
  date 		= 	{2025-05-27},
}

@book{Ruggeri2021,
  title        = {Classical and Relativistic Rational Extended Thermodynamics of Gases},
  author       = {Ruggeri, T. and Sugiyama, M.},
  publisher    = {Springer Cham},
  year         = {2021},
  edition      = {1},
  doi          = {10.1007/978-3-030-59144-1},
}

@article{Grifo2025II,
  title		=	{Travelling waves in dryland ecology: continuous and discontinuous connections in a hyperbolic vegetation model},
  author	=	{Grif\`o, G. and Curr\`o, C. and Valenti, G.},
  journal	=	{Nonlinear Dynamics},
  year		=	{2025},
  volume  	=       {113},
  pages 	= 	{15295–15319},
  DOI 		= 	{10.1007/s11071-025-10925-7},
}

@article{Consolo2025,
  title		=	{Vegetation pattern formation and transition in dryland ecosystems with finite soil resources and inertia},
  author	=	{Consolo, G. and Curr\`o, C. and Grif\`o, G. and Valenti, G.},
  journal	=	{Physica D: Nonlinear Phenomena},
  year		=	{2025},
  volume 	= 	{476},
  pages 	= 	{134601},
  DOI 		= 	{10.1016/j.physd.2025.134601},
}

@article{Consolo2024II,
title = {Stationary and Oscillatory patterned solutions in three-compartment reaction-diffusion systems: Theory and application to dryland ecology},
journal = {Chaos, Solitons \& Fractals},
volume = {186},
pages = {115287},
year = {2024},
doi = {10.1016/j.chaos.2024.115287},
author = {Consolo, G. and Curr\`o, C. and Grif\`o, G. and Valenti, G.},
}

@article{BARBERA2015,
title = {On discontinuous travelling wave solutions for a class of hyperbolic reaction–diffusion models},
journal = {Physica D: Nonlinear Phenomena},
volume = {308},
pages = {116-126},
year = {2015},
issn = {0167-2789},
doi = {10.1016/j.physd.2015.06.011},
author = {E. Barbera and C. Currò and G. Valenti},
}

@article{EIGENTLER2020b,
title = {Effects of precipitation intermittency on vegetation patterns in semi-arid landscapes},
journal = {Physica D: Nonlinear Phenomena},
volume = {405},
pages = {132396},
year = {2020},
issn = {0167-2789},
doi = {10.1016/j.physd.2020.132396},
author = {L. Eigentler and J.A. Sherratt},
}

@article{EIGENTLER2024,
title = {Delayed loss of stability of periodic travelling waves: {I}nsights from the analysis of essential spectra},
journal = {Journal of Theoretical Biology},
volume = {595},
pages = {111945},
year = {2024},
issn = {0022-5193},
doi = {10.1016/j.jtbi.2024.111945},
author = {Eigentler, L. and Sensi, M.},
}

@article{SHERRATT2013,
title = {Pattern solutions of the {K}lausmeier model for banded vegetation in semi-arid environments III: The transition between homoclinic solutions},
journal = {Physica D: Nonlinear Phenomena},
volume = {242},
number = {1},
pages = {30-41},
year = {2013},
issn = {0167-2789},
doi = {10.1016/j.physd.2012.08.014},
author = {J.A. Sherratt},
}

@article{Sherratt2012,
  title     = {Vegetation patterns and desertification waves in semi-arid environments: {M}athematical models based on local facilitation in plants},
  author    = {Sherratt, J.A. and Synodinos, A.D.},
  journal   = {Discrete and Continuous Dynamical Systems - B},
  volume    = {17},
  number    = {8},
  pages     = {2815--2827},
  year      = {2012},
  doi       = {10.3934/dcdsb.2012.17.2815}
}

@article{VanSaarloosHohenberg1992,
  title     = {Fronts, pulses, sources and sinks in generalized complex {G}inzburg-{L}andau equations},
  author    = {van Saarloos, W. and Hohenberg, P. C.},
  journal   = {Physica D},
  volume    = {56},
  pages     = {303--367},
  year      = {1992},
  DOI = {10.1016/0167-2789(92)90175-M}
}

@article{Doelman1995,
  title     = {Instability of quasiperiodic solutions of the {G}inzburg-{L}andau equation},
  author    = {Doelman, A. and Gardner, R. and Jones, C.},
  journal   = {Proceedings of the Royal Society of Edinburgh: Section A Mathematics},
  volume    = {125},
  pages     = {501--526},
  year      = {1995},
  DOI ={10.1017/S0308210500032649}
}

@article{Mielke1998,
  title     = {Bounds for the solutions of the complex {G}inzburg-{L}andau equation in terms of the dispersion parameters},
  author    = {Mielke, A.},
  journal   = {Physica D},
  volume    = {117},
  pages     = {106--116},
  year      = {1998},
  DOI ={10.1016/S0167-2789(97)00308-4}
}

@incollection{Mielke2002,
  title     = {The {Ginzburg-Landau} equation in its role as a modulation equation},
  author    = {Mielke, A.},
  booktitle = {Handbook of Dynamical Systems},
  editor    = {Fiedler, B.},
  publisher = {Elsevier},
  address   = {Amsterdam},
  volume    = {2},
  pages     = {759--834},
  year      = {2002}
}

@article{AransonKramer2002,
  title     = {The world of the complex {G}inzburg-{L}andau equation},
  author    = {Aranson, I. S. and Kramer, L.},
  journal   = {Reviews of Modern Physics},
  volume    = {74},
  pages     = {99--143},
  year      = {2002},
  DOI = {10.1103/RevModPhys.74.99},
}

@article{Deblauwe2011,
  title		=	{Environmental modulation of self-organized periodic vegetation patterns in {S}udan},
  author	=	{Deblauwe, V. and Couteron, P. and Lejeune, O. and Bogaert, J. and Barbier, N.},
  journal	=	{Ecography},
  volume	=	{34},
  number	=	{6},
  pages		=	{990--1001},
  year		=	{2011},
  DOI 		=	{10.1111/j.1600-0587.2010.06694.x},
}

@article{Deblauwe2012,
  title		=	{Determinants and dynamics of banded vegetation pattern migration in arid climates},
  author	=	{Deblauwe, V. and Couteron, P. and Bogaert, J. and Barbier, N.},
  journal	=	{Ecological Monograph},
  volume	=	{82},
  number	=	{1},
  pages		=	{3--21},
  year		=	{2012},
  DOI 		=	{10.1890/11-0362.1},
}

@article{Carter2024,
  title={Travelling pulses on three spatial scales in a {Klausmeier-type} vegetation-autotoxicity model},
  author={Carter, P. and Doelman, A. and Iuorio, A. and Veerman, F.},
  journal={Nonlinearity},
  volume={37},
  number={9},
  pages={095008},
  year={2024},
  publisher={IOP Publishing},
  DOI = {10.1088/1361-6544/ad6112}
}

@article{Giunta2021,
author = {Giunta, V. and Lombardo, M.C. and Sammartino, M.},
title = {Pattern Formation and Transition to Chaos in a Chemotaxis Model of Acute Inflammation},
journal = {SIAM Journal on Applied Dynamical Systems},
volume = {20},
number = {4},
pages = {1844-1881},
year = {2021},
doi = {10.1137/20M1358104},
}

@article{Wollkind1994,
author = {Wollkind, D.J. and Manoranjan, V.S. and Zhang, L.},
title = {Weakly Nonlinear Stability Analyses of Prototype Reaction-Diffusion Model Equations},
journal = {SIAM Review},
volume = {36},
number = {2},
pages = {176-214},
year = {1994},
doi = {10.1137/1036052},
}

\end{document}